\documentclass[11pt]{article}

\usepackage{graphicx,amssymb,amsmath}
\usepackage{amsfonts}
\usepackage{booktabs}
\usepackage{cite}
\usepackage{color}
\usepackage{enumerate}
\usepackage{float}

\usepackage[pdftex, plainpages = false, pdfpagelabels,
                 bookmarks=false,
                 bookmarksopen = true,
                 bookmarksnumbered = true,
                 breaklinks = true,
                 linktocpage,
                 pagebackref,
                 colorlinks = true,  
                 linkcolor = blue,
                 urlcolor  = cyan,
                 citecolor = red,
                 anchorcolor = green,
                 hyperindex = true,
                 hyperfigures
                 ]{hyperref}

\usepackage{mathrsfs}
\usepackage{subfig}
\usepackage{tabularx}
\usepackage{authblk}

\usepackage[in]{fullpage}
\usepackage[top=0.8in, bottom=0.9in, left=0.9in, right=0.9in]{geometry}
\def\bbR{\mathbb{R}}
\def\calA{\mathcal{A}}
\def\calD{\mathcal{D}}
\newtheorem{observation}{Observation}

\newtheorem{corollary}{Corollary}

\newtheorem{lemma}{Lemma}
\newtheorem{theorem}{Theorem}
\newenvironment{proof}{\noindent {\textbf{Proof:}}\rm}{\hfill $\Box$ \rm\bigskip}

\title{A Deterministic Partition Tree and Applications\thanks{A preliminary version of this paper will appear in {\em Proceedings of the 33rd Annual European Symposium on Algorithms (ESA 2025)}.}
}

\author{
Haitao Wang
}
\affil{Kahlert School of Computing\\
University of Utah, Salt Lake City, UT 84112, USA
\\ {\tt haitao.wang@utah.edu}}

\begin{document}

\pagestyle{plain}
\pagenumbering{arabic}
\setcounter{page}{1}
\date{}

\thispagestyle{empty}
\maketitle

\vspace{-0.35in}
\begin{abstract}
In this paper, we present a deterministic variant of Chan's randomized partition tree [Discret. Comput. Geom., 2012].
This result leads to numerous applications. In particular, for $d$-dimensional simplex range counting (for any constant $d \ge 2$), we construct a data structure using $O(n)$ space and $O(n^{1+\epsilon})$ preprocessing time, such that each query can be answered in $o(n^{1-1/d})$ time (specifically, $O(n^{1-1/d} / \log^{\Omega(1)} n)$ time), thereby breaking an $\Omega(n^{1-1/d})$ lower bound known for the semigroup setting.
Notably, our approach does not rely on any bit-packing techniques. 
We also obtain deterministic improvements for several other classical problems, including simplex range stabbing counting and reporting, segment intersection detection, counting and reporting, ray-shooting among segments, and more.
Similar to Chan's original randomized partition tree, we expect that additional applications will emerge in the future, especially in situations where deterministic results are preferred.
\end{abstract}



\section{Introduction}
\label{sec:intro}

Simplex range searching is a fundamental problem in computational geometry. Given a set $P$ of $n$ points in the $d$-dimensional space $\bbR^d$ for a constant $d\geq 2$, the goal is to build a data structure so that points of $P$ inside a query simplex can be found efficiently. The problem has been extensively studied (see \cite{ref:AgarwalRa17,ref:AgarwalSi17,ref:MatousekGe94} for some excellent surveys).
For solving the problem with small space (e.g., near linear), one powerful technique is {\em partition trees}, e.g., ~\cite{ref:ChanOp12,ref:EdelsbrunnerHa86,ref:HausslerEp87,ref:MatousekEf92,ref:MatousekRa93,ref:WillardPo82,ref:YaoA83,ref:YaoPa89}. In particular, using a partition tree Matou\v{s}ek~\cite{ref:MatousekEf92} built a data structure of $O(n)$ space in $O(n\log n)$ time and each simplex range query can be answered in $O(n^{1-1/d}\log^{O(1)} n)$ time. Subsequently Matou\v{s}ek~\cite{ref:MatousekRa93} gave another more complicated data structure of $O(n)$ space with $O(n^{1+\epsilon})$ preprocessing time and $O(n^{1-1/d})$ query time; throughout the paper let $\epsilon$ represent an arbitrarily small positive constant. Chazelle~\cite{ref:ChazelleLo89} proved that $\Omega(n^{1-1/d}/\log n)$ (and $\Omega(\sqrt{n})$ for $d=2$) is a lower bound on the query time for an $O(n)$-space data structure; it is widely believed that the $\log n$ factor is an artifact of the proof. Although the result of \cite{ref:MatousekRa93} seems to achieve optimal query time with linear space, it is not quite satisfactory. One reason is that partition trees are often used as backbone for designing multi-level data structures and certain properties of the result of \cite{ref:MatousekRa93} makes this challenging, e.g., it has a special root of $O(n^{1/d}\log n)$ degree whose children may have overlapping cells and it does not guarantee the optimal crossing number except at the bottom most level of the tree. To address these issues,
Chan~\cite{ref:ChanOp12} later proposed a randomized partition tree of $O(n)$ space that can be built in $O(n\log n)$ expected time and the query time is bounded by $O(n^{1-1/d})$ w.h.p.\footnote{With high probability, i.e., with probability at least $1-1/n^{\delta}$ for an arbitrarily large constant $\delta$.}  Comparing to Matou\v{s}ek's result~\cite{ref:MatousekRa93}, Chan's partition tree has many nice properties, e.g., each node has $O(1)$ children, crossing number is optimal (with w.h.p) at almost all layers (except the top few layers), and children’s cells at each node are pairwise disjoint. These properties make Chan's partition tree quite amenable to multi-level data structures~\cite{ref:ChanOp12,ref:ChanHo23,ref:ChanSi23}.

Simplex range searching has several versions: (1) Counting: compute the number of points of $P$ inside the query simplex $\Delta$; (2) reporting: report all points of $P$ inside $\Delta$; (3) semigroup query (which generalizes the counting query): compute the sum of the weights of the points of $P$ inside $\Delta$, assuming that each point of $P$ is assigned a weight from a semigroup. The above results~\cite{ref:ChanOp12,ref:MatousekEf92,ref:MatousekRa93} are applicable to semigroup queries and can also be modified to solve range reporting with an additive term $k$ in the query time, where $k$ is the output size. In particular, Chazelle's lower bound~\cite{ref:ChazelleLo89} is for the semigroup setting.

Since Chan's partition tree is randomized, for those who need deterministic algorithms, Matou\v{s}ek's partition trees~\cite{ref:MatousekEf92,ref:MatousekRa93} are still the main resort. In this paper, we ``partially'' derandomize Chan's partition tree and obtain a deterministic tool, especially for designing multi-level data structures. More specifically, our partition tree is similar to Chan's (e.g., $O(n)$ space, optimal crossing number at all levels except the top few levels, disjointness of children's cells of each node); however, the degree of each node is logarithmic instead of constant (that is why we used ``partially'' above).
Albeit this drawback, our tree is still powerful enough to have many applications, as discussed below.


\paragraph{Simplex range counting.}
For simplex range counting, we construct a data structure of $O(n)$ space that can answer each query in $O(n^{1-1/d}/\log^{\Omega(1)} n)$ time.
Note that this does not violate Chazelle's lower bound~\cite{ref:ChazelleLo89} as it is for the more general semigroup queries. The preprocessing time is $O(n^{1+\epsilon})$.

To achieve the result, we construct our partition tree so that each leaf has $O(\log^{\tau} n)$ points of $P$ for an arbitrarily small constant $\tau>0$. Because this number is small, we can afford to preprocess all leaves in $O(n)$ time by considering all possible configurations for queries; in this way, each query on a single leaf $v$ can be answered in $O(\log\log n)$ time. While this kind of technique may not be quite surprising, it has never been used in simplex range searching, perhaps because previous work has been focusing on the semigroup queries while this technique does not work in that setting.
Note that our above result is also applicable to simplex range emptiness queries.

It should be noted that very recently Chan and Zheng~\cite{ref:ChanSi23} obtained similar randomized results (i.e., $O(n)$ words of space and $o(n^{1-1/d})$ query time w.h.p.) for other data structures and also mentioned a possibility of achieving such result for simplex range counting. One difference is that their technique uses 
bit-packing by assuming each word has $\Omega(\log n)$ bits (note that the $\log n$-bit word RAM is also a conventional computational model), while ours does not use bit-packing. 
In addition, their result is randomized while ours is deterministic.

\paragraph{Simplex range stabbing.}
Given a set of $n$ simplices in $\bbR^d$, the {\em simplex range stabbing counting problem} is to build a data structure to compute the number of simplices containing a query point. Using Chan's partition tree, Chan and Zheng~\cite{ref:ChanSi23} built a randomized data structure of $O(n\log\log n)$ space in $O(n\log n)$ expected preprocessing time and the query time is $O(n^{1-1/d})$ w.h.p. Using our new  deterministic partition tree, we build a deterministic data structure of $O(n\log\log n)$ space in $O(n^{1+\epsilon})$ preprocessing time and the query time is $O(n^{1-1/d}/\log^{\Omega(1)}n)$. Previously, the best deterministic results~\cite{ref:MatousekEf92,ref:MatousekRa93} have $O(n^{1+\epsilon})$ preprocessing time, $O(n\log^{O(1)}n)$ space, and $O(n^{1-1/d}\log^{O(1)}n)$ query time; or $O(n2^{\sqrt{\log n}})$ preprocessing time and space, and $n^{1-1/d}\cdot 2^{O(\sqrt{\log n})}$ query time.

For the {\em reporting problem} (i.e., report all simplices containing a query point), the randomized data structure of \cite{ref:ChanSi23} has the same performance as above except that the query time becomes $O(k+n^{1-1/d})$ w.h.p., where $k$ is the output size. We also obtain the same deterministic result as above with $O(k+n^{1-1/d}/\log^{\Omega(1)} n)$ query time.

Using bit-packing tricks, Chan and Zhen~\cite{ref:ChanSi23} further reduce the space of their data structure to $O(n)$ words of space. Again, all our results in this paper use the conventional computational model, without allowing bit-packing techniques.

\paragraph{Segment intersection searching.}
Given a set of $n$ (possibly intersecting) line segments in the plane, the {\em segment intersection counting problem} is to build a data structure to compute the number of segments intersecting a query segment. Chan and Zhen~\cite{ref:ChanSi23} built a randomized data structure of $O(n\log\log n)$ space (which again can be reduced to $O(n)$ words of space if bit-packing tricks are allowed) in $O(n\log n)$ expected preprocessing time and the query time is $O(\sqrt{n})$ w.h.p. Using our new deterministic partition tree, we build a deterministic data structure of $O(n\log\log n)$ space in $O(n^{1+\epsilon})$ preprocessing time and the query time is $O(\sqrt{n}/\log^{\Omega(1)} n)$.
The previously best deterministic result~\cite{ref:Bar-YehudaVa94} built a data structure of $O(n\log^2 n)$ space in $O(n^{3/2})$ time that can answer each query in $O(\sqrt{n}\log n)$ time.

For the {\em reporting problem} (i.e., report all segments intersecting a query segment), the randomized data structure of \cite{ref:ChanSi23} has the same performance as above except that the query time becomes $O(k+\sqrt{n})$ w.h.p., where $k$ is the output size. We also obtain the same deterministic result as above with $O(k+\sqrt{n}/\log^{\Omega(1)} n)$ query time.



\paragraph{Segment intersection detection.}
Given a set of $n$ (possibly intersecting) line segments in the plane, the problem is to build a data structure to determine whether a query line intersects any segment~\cite{ref:ChengAl92,ref:WangAl20}. The previously best deterministic result~\cite{ref:WangAl20} builds an $O(n)$-space data structure in $O(n^{3/2})$ time and each query can be answered in $O(\sqrt{n}\log n)$ time. Using our new partition tree, we build an $O(n)$-space data structure with $O(\sqrt{n}/\log^{\Omega(1)} n)$ query time and $O(n^{1+\epsilon})$ preprocessing time.

\paragraph{Ray-shooting among non-intersecting segments.}
Given a set of $n$ non-intersecting line segments in the plane, the problem is to build a data structure to find the first segment hit by a query ray~\cite{ref:AgarwalRa93,ref:Bar-YehudaVa94,ref:OvermarsSt90}.  The previously best deterministic result~\cite{ref:WangAl20} builds an $O(n)$-space data structure in $O(n^{3/2})$ time and each query can be answered in $O(\sqrt{n}\log n)$ time. Using our new partition tree, we build an $O(n)$-space data structure with $O(\sqrt{n}/\log^{\Omega(1)} n)$ query time and $O(n^{1+\epsilon})$ preprocessing time.

If the segments are allowed to intersect, then the problem has also been studied~\cite{ref:AgarwalRa93,ref:AgarwalAp93,ref:Bar-YehudaVa94,ref:ChanOp12,ref:ChengAl92,ref:GuibasIn88,ref:OvermarsSt90}. The previously best deterministic result~\cite{ref:WangAl20} builds an $O(n\log n)$-space data structure in $O(n^{3/2})$ time and each query can be answered in $O(\sqrt{n}\log n)$ time. Chan and Zheng~\cite{ref:ChanSi23} built a randomized data structure of $O(n\log\log n)$ space (which again can be reduced to $O(n)$ words of space if bit-packing tricks are allowed) in $O(n\log n)$ expected preprocessing time and the query time is $O(\sqrt{n})$ w.h.p.



\paragraph{Outline.}
The rest of the paper is organized as follows. After introducing a basic tool (i.e., cuttings) in Section~\ref{sec:pre}, we present our partition tree in Section~\ref{sec:partitiontree}. The simplex range counting problem is treated in Section~\ref{sec:simcount}. The simplex range stabbing and the segment intersection searching problems are discussed in Section~\ref{sec:simstab}. Subsequenctly, we solve the segment intersection detection and the ray-shooting problems in Sections~\ref{sec:segdetect} and \ref{sec:rayshoot}, respectively.

\section{Preliminaries}
\label{sec:pre}

Let $H$ be a set of $n$ hyperplanes in $\bbR^d$. We use $\calA(H)$ to denote the arrangement of $H$, which can be computed in $O(n^d)$ time~\cite{ref:EdelsbrunnerCo86}.
For a compact region $R\in \bbR^d$, we use
$H_R$ to denote the subset of hyperplanes of $H$ that intersect the relative interior of $R$ but does not contain $R$ (we also say that these hyperplanes {\em cross} $R$).
A {\em cutting} for $H$ is a collection $\Xi$ of closed cells (each of which is a simplex, possibly unbounded) with disjoint interiors, which together cover the entire space $\bbR^d$~\cite{ref:ChazelleCu93,ref:MatousekRa93}.
The {\em size} of $\Xi$ is the number of cells of $\Xi$. For a parameter $r$ with $1\leq r\leq n$, a {\em $(1/r)$-cutting} for $H$ is a cutting $\Xi$ satisfying $|H_{\sigma}|\leq n/r$ for every cell $\sigma\in \Xi$.



For any $1\leq r\leq n$, a $(1/r)$-cutting of size $O(r^d)$ for $H$ can be computed in $O(nr^{d-1})$ time~\cite{ref:ChazelleCu93}. We further have the following lemma (similar results were mentioned before~\cite{ref:AgarwalPs05,ref:ChanOp12,ref:WangUn23}).

\begin{lemma}\label{lem:cutting}{\em (Cutting Lemma)}
Let $H$ be a set of $n$ hyperplanes and $\Delta$ a simplex in $\bbR^d$. For any $1\leq r\leq n$, we can compute a $(1/r)$-cutting of $O(K\cdot (r/n)^d+r^{d-1+\beta})$ cells for $H$ whose union is $\Delta$ in $O(K\cdot (r/n)^{d-1}+nr^{d-2+\beta})$ time, where $K$ is the number of vertices of $\calA(H)$ inside $\Delta$ and $\beta$ is an arbitrarily small positive constant.
\end{lemma}
\begin{proof}
We apply Chazelle's algorithm~\cite{ref:ChazelleCu93} but only on the region inside $\Delta$ (i.e., starting with $C_0=\Delta$ following the notation of \cite{ref:ChazelleCu93}). A detailed analysis for the 2D case is given in~\cite[Appendix A]{ref:WangUn23}. Below we sketch how to modify the analysis in~\cite{ref:ChazelleCu93} accordingly by following the notation there.

The analysis follows the same approach except that we use $K$ to replace $\binom{n}{d}$ in the formula $\sum_{s\in C_{k-1}}v(H_{|s};s)\leq \binom{n}{d}$ in \cite[Page 153]{ref:ChazelleCu93}. As such, the subsequent formula becomes
$$|C_k|\leq c\left(\frac{r_0^k\log r_0}{n}\right)^d\cdot K+cr_0^{d-1}(\log r_0)^d|C_{k-1}|.$$
Then, one can prove by induction that $|C_k|\leq r_0^{d(k+1)}\cdot K/n^d+r_0^{(k+1)\cdot (d-1+\beta)}$, for an arbitrarily small constant $\beta>0$.
Therefore, the total number of cells of the cutting inside $\Delta$ is as stated in the lemma.

For the time analysis, following the same formula $\sum_{0\leq k\leq \lceil \log_{r_0}r\rceil}\frac{n}{r_0^k}|C_k|$ in \cite[Page 153]{ref:ChazelleCu93} and using the above inequality for $|C_k|$, we can derive the time complexity as stated in the lemma. 
\end{proof}

Throughout the paper, $\beta$ always refers to the one in the above lemma.

\section{Deterministic partition tree}
\label{sec:partitiontree}

In this section, we present our deterministic partition tree.
We start with the following lemma, which derandomizes Chan's partition refinement theorem (i.e., Theorem 3.1~\cite{ref:ChanOp12}).

\begin{lemma}\label{lem:partition}
Let $P$ be a set of $n$ points and $H$ a set of $m$ hyperplanes in $\bbR^d$. Suppose there are $t$ interior-disjoint cells whose union covers $P$, such that each cell contains at most $2n/t$ points of $P$ and each hyperplane crosses at most $\kappa$ cells. Then, for any $b\geq 4$, we can divide every cell into $O(b)$ disjoint subcells each containing at most $2n/(bt)$ points of $P$, for a total of at most $bt$ subcells, such that the total number of subcells crossed by any hyperplane in $H$ is bounded by
\begin{align}\label{equ:bound}
    O((b\cdot t)^{1-1/d}+b^{1-1/(d-1+\beta)}\cdot \kappa\cdot \log t + b\cdot \log m).
\end{align}
\end{lemma}

\paragraph{Remark.}
Comparing to Chan's partition refinement theorem, there is an extra $\log t$ in the second term of \eqref{equ:bound}. As will be seen later, due to this extra factor, we have to make each node of our partition tree have logarithmically many children instead of constant.

\subsection{Proving Lemma~\ref{lem:partition}}
This subsection is devoted to the proof of Lemma~\ref{lem:partition}.

Let $S$ denote the set of all $t$ given cells in Lemma~\ref{lem:partition}.
Let $H'$ be a multiset containing $C$ copies of each hyperplane in $H$, for a parameter $C$ that is a sufficiently large power of $b$ (so that all future multiplicities are integers; the actual value of $C$ is not important). For any multiset $H''$ of $H$,
the {\em size} of $H''$, denoted by $|H''|$, is the sum of the multiplicities of the hyperplanes in $H''$. For a cell $\Delta$, let $N_{\Delta}(H'')$ denote the number of vertices of the arrangement $\calA(H'')$ of $H''$ inside $\Delta$, counting multiplicities (note that the multiplicity of a vertex is the product of the multiplicities of its defining hyperplanes); let $H''(\Delta)$ denote the (multi-)subset of hyperplanes of $H''$ crossing $\Delta$.

We process the $t$ cells of $S$ iteratively one by one. The order they will be processed is carefully chosen (in contrast, a random order is used in the proof of \cite{ref:ChanOp12}, which is a major difference between our proof and that in \cite{ref:ChanOp12}). We will assign indices to the cells following the reverse order they are processed.
Suppose we have processed $t-i$ cells, which have been assigned indices and denoted by $\Delta_t,\Delta_{t-1},\ldots,\Delta_{i+1}$.
In the next iteration, we will find a cell among the unprocessed cells of $S$ and process it (and the cell will be assigned index $i$ as $\Delta_i$).
Suppose we now have a multiset $H_i$ after cell $\Delta_{i+1}$ is processed (initially let $H_{t+1}=H'$).
Let $S_i$ denote the subset of unprocessed cells of $S$. Hence, $|S_i|=i$.

For each cell $\Delta\in S_i$, define
$$A_{\Delta}=\left(\frac{N_{\Delta}(H_i)}{b}\right)^{1/d},\ \  B_{\Delta}=\frac{|H_i(\Delta)|}{b^{1/(d-1+\beta)}}.$$
Define $S_{i1}$ as the subset of cells $\Delta\in S_i$ with $A_{\Delta}\geq B_{\Delta}$. Let $S_{i2}=S\setminus S_{i1}$. If $|S_{i1}|\geq i/2$, then we define $\Delta_i$ as the cell $\Delta$ of $S_{i1}$ that minimizes the value $A_{\Delta}$; otherwise define $\Delta_i$ as the cell $\Delta$ of $S_{i2}$ that minimizes $B_{\Delta}$. Note that the above way of defining $\Delta_i$ is a key difference from Chan's approach~\cite{ref:ChanOp12}, where $\Delta_i$ is chosen from $S_i$ randomly.
We now process $\Delta_i$ in the following three steps (which is similar to Chan's approach).

\begin{enumerate}
\item
Construct a $(1/r_i)$-cutting for $H_i$ inside $\Delta_i$ with
$$r_i=c\cdot \min\left\{|H_i(\Delta_i)|\cdot \left(\frac{b}{N_{\Delta_i}(H_i)}\right)^{1/d},b^{1/(d-1+\beta)}\right\},$$
for some constant $c$. By the cutting lemma, the number of subcells inside $\Delta_i$ is $O(N_{\Delta_i}(H_i)\cdot (r_i/|H_i(\Delta_i)|)^d+r_i^{d-1+\beta})$, which can be made at most $b/4$ for a sufficiently small $c$.

\item
We further subdivide each subcell of $\Delta_i$ (e.g., using vertical cuts) so that each subcell contains at most $2n/(tb)$ points of $P$. The number of extra cuts is $O(b)$ as $\Delta_i$ contains at most $2n/t$ points of $P$. The total number of extra cuts for processing all $t$ cells of $S$ is at most $\frac{n}{2n/(tb)}+t=bt/2+t$, and thus the total number of subcells after processing all $t$ cells is at most $bt/4 + bt/2+t=3bt/4+t$, which is at most $bt$ for $b\geq 4$.

\item
For each distinct hyperplane $h\in H_i$, multiply the multiplicity of $h$ in $H_i$ by $(1+1/b)^{\lambda_i(h)}$, where $\lambda_i(h)$ is the number of subcells of $\Delta_i$ crossed by $h$.
Let $H_{i-1}$ be the resulting multiset after this.
\end{enumerate}

To prove Lemma~\ref{lem:partition}, it remains to prove Bound \eqref{equ:bound}.

In the third step, since each of the $O(b)$ subcells of $\Delta_i$ is crossed by at most $|H_i(\Delta_i)|/r_i$ hyperplanes of $H_i$, we have $\sum_{h\in H_i}\lambda_i(h)=O(b\cdot |H_i(\Delta_i)|/r_i)$. Since $\lambda_i(h)=O(b)$, we have
\begin{align*}
    |H_{i-1}| = \sum_{h\in H_i}(1+1/b)^{\lambda_i(h)}=\sum_{h\in H_i}(1+O(\frac{\lambda_i(h)}{b})) = |H_i| + \sum_{h\in H_i}O(\frac{\lambda_i(h))}{b}).
\end{align*}
We thus obtain
\begin{align*}
|H_{i-1}|-|H_i| =O(\frac{1}{b}\cdot \sum_{h\in H_i}\lambda_i(h)) = O(|H_i(\Delta_i)|/r_i) = O(\alpha_i)\cdot |H_i|,
\end{align*}
where
$$\alpha_i=\frac{|H_i(\Delta_i)|}{r_i\cdot |H_i|}=O\left(\frac{1}{|H_i|}\cdot \max\left\{\left(\frac{N_{\Delta_i}(H_i)}{b}\right)^{1/d},\frac{|H_i(\Delta_i)|}{b^{1/(d-1+\beta)}}\right\}\right)=O\left(\frac{1}{|H_i|}\cdot \max\left\{A_{\Delta_i},B_{\Delta_i}\right\}\right).$$

After all $t$ cells of $S$ are processed, we have a multiset $H_0$. Recall that $H_{t+1}=H'$ and $|H'|=C\cdot m$. According to the above analysis, we have
\begin{align}\label{equ:20}
|H_0|=|H_{t+1}|\cdot \Pi_{i=1}^t(1+O(\alpha_i))\leq C\cdot m\cdot \exp\left(O(\sum_{i=1}^t\alpha_i)\right).
\end{align}

The following lemma gives an upper bound for $\sum_{i=1}^t\alpha_i$.

\begin{lemma}\label{lem:30}
$\sum_{i=1}^t\alpha_i=O(\frac{t^{1-1/d}}{b^{1/d}}+\frac{\kappa\cdot \log t}{b^{1/(d-1+\beta)}})$.
\end{lemma}
\begin{proof}
Recall that each cell $\Delta_i$ is either from $S_{i1}$ or from $S_{i2}$.
Define $I_1$ (resp., $I_2$) as the set of indices $i\in [1,t]$ such that cell $\Delta_i$ is from $S_{i1}$ (resp., $S_{i2}$).
In what follows, we first prove $\sum_{i\in I_1}\alpha_i=O(\frac{t^{1-1/d}}{b^{1/d}})$ and then prove
$\sum_{i\in I_2}\alpha_i=O(\frac{\kappa\cdot \log t}{b^{1/(d-1+\beta)}})$, which will lead to the lemma.

\paragraph{Proving $\boldsymbol{\sum_{i\in I_1}\alpha_i=O(\frac{t^{1-1/d}}{b^{1/d}})}$.}
For each $i\in I_1$, since $\Delta_i$ is from $S_{i1}$, by the definition of $S_{i1}$, we have $A_{\Delta_i}\geq B_{\Delta_i}$. Therefore, $\alpha_i=(N_{\Delta_i}(H_i)/b)^{1/d}/|H_i|$. By definition, $\Delta_i$ is a cell $\Delta$ of $S_{i1}$ minimizing $N_{\Delta}(H_i)$ and $|S_{i1}|\geq i/2$. Because cells of $S_{i}$ are interior-disjoint, we have
$$N_{\Delta_i}(H_i)\leq \frac{\sum_{\Delta\in S_{i1}}N_{\Delta}(H_i)}{i/2}=O\left(\frac{|H_i|^d}{i}\right).$$
As such, we have
$$\alpha_i=\frac{(N_{\Delta_i}(H_i)/b)^{1/d}}{|H_i|}=O\left(\left(\frac{1}{b\cdot i}\right)^{1/d}\right),$$
and
$$\sum_{i\in I_1}\alpha_i=O\left(\sum_{i\in I_1}\left(\frac{1}{b\cdot i}\right)^{1/d}\right)=O\left(\frac{1}{b^{1/d}}\cdot \sum_{i=1}^t\frac{1}{i^{1/d}}\right).$$

It can be proved (e.g., by induction) that $\sum_{i=1}^t\frac{1}{i^{1/d}}\leq 2\cdot t^{1-1/d}$ holds for $d\geq 2$. Hence, $\sum_{i\in I_1}\alpha_i=O(\frac{t^{1-1/d}}{b^{1/d}})$.

\paragraph{Proving $\boldsymbol{\sum_{i\in I_2}\alpha_i=O(\frac{\kappa\cdot \log t}{b^{1/(d-1+\beta)}})}$.}
For each $i\in I_2$, since $\Delta_i$ is from $S_{i2}$, by the definition of $S_{i2}$, we have $A_{\Delta_i}< B_{\Delta_i}$. Therefore, $\alpha_i=(1/|H_i|)\cdot |H_i(\Delta_i)|/b^{1/(d-1+\beta)}$. By definition, $\Delta_i$ is a cell $\Delta$ of $S_{i2}$ minimizing $|H_i(\Delta)|$ and $|S_{i2}|\geq i/2$. Hence, we can obtain
$$|H_i(\Delta_i)|\leq \frac{1}{|S_{i2}|}\cdot \sum_{\Delta\in S_{i2}}|H_i(\Delta)|\leq \frac{2}{i}\cdot \sum_{\Delta\in S_{i2}}|H_i(\Delta)|\leq \frac{2}{i}\cdot \sum_{\Delta\in S}|H_i(\Delta)|.$$
Since each hyperplane of $H$ crosses at most $\kappa$ cells of $S$, we have $\sum_{\Delta\in S}|H_i(\Delta)|\leq \kappa\cdot |H_i|$. We thus obtain $|H_i(\Delta_i)|\leq 2\cdot \kappa\cdot |H_i|/i$, and $\alpha_i\leq 2\cdot \kappa/(i\cdot b^{1/(d-1+\beta)})$, which leads to $$\sum_{i\in I_2}\alpha_i\leq \sum_{i\in I_2}\frac{2\kappa}{i\cdot b^{1/(d-1+\beta)}}=\frac{2\kappa}{b^{1/(d-1+\beta)}}\cdot \sum_{i\in I_2}\frac{1}{i}\leq \frac{2\kappa}{b^{1/(d-1+\beta)}}\cdot \sum_{i=1}^t\frac{1}{i}=O\left(\frac{\kappa\cdot \log t}{b^{1/(d-1+\beta)}}\right).$$
The lemma thus follows.
\end{proof}

By Lemma~\ref{lem:30} and \eqref{equ:20}, we have
\begin{equation}\label{equ:25}
|H_0|\leq C\cdot m\cdot \exp\left(O\left(\frac{t^{1-1/d}}{b^{1/d}}+\frac{\kappa\cdot \log t}{b^{1/(d-1+\beta)}}\right)\right).    \end{equation}
For any hyperplane $h\in H$, let $\lambda(h)$ be the total number of subcells crossed by $h$. Our goal is to prove that $\lambda(h)$ is bounded by \eqref{equ:bound}. By the way $H_0$ is produced, we have $C\cdot (1+1/b)^{\lambda(h)}\leq |H_0|$. Hence, $$\lambda(h)\leq \log_{1+1/b}\frac{|H_0|}{C}=O(b\cdot \log \frac{|H_0|}{C})=O(b\log m+ (bt)^{1-1/d}+b^{1-1/(d-1+\beta)}\cdot \kappa\cdot \log t).$$

This proves Lemma~\ref{lem:partition}.

\subsection{Constructing the partition tree}

We now construct the partition tree using Lemma~\ref{lem:partition}. Our main goal is to prove the following lemma.

\begin{lemma}\label{lem:partitiontree}
Given a set $P$ of $n$ points, a set $H$ of $m=n^{O(1)}$ hyperplanes in $\bbR^d$, and a parameter $r$ with $r_0\leq r\leq n$ for a sufficiently large constant $r_0$, for $b=\log^\rho r$ for a sufficiently large constant $\rho>0$, there exists a sequence  $\Pi_0,\Pi_1,\ldots,\Pi_{k+1}$ with $k=\lfloor \log_b r\rfloor$ and $b'=\lceil r/b^{k}\rceil$, where each $\Pi_i$ is a collection of disjoint simplicial cells (i.e., each cell is a simplex) whose union covers $P$ with the following properties:
\begin{enumerate}
  \item $\Pi_0$ has only one cell that is $\bbR^d$ and the number of cells of $\Pi_i$ is $O(b'\cdot b^{i-1})$ for $i\geq 1$ (in particular, $\Pi_{k+1}$ has $O(r)$ cells; the total number of cells in all collections is also $O(r)$).
  \item Each cell of $\Pi_i$ contains at least one point and at most $2n/(b'\cdot b^{i-1})$ points of $P$ for all $i\geq 1$ (in particular, each cell of $\Pi_{k+1}$ contains at most $2n/r$ points).
  \item Each cell in $\Pi_{i+1}$ is contained in a single cell of $\Pi_i$.
  \item Each cell of $\Pi_i$ has $O(b)$ cells of $\Pi_{i+1}$ for $i\geq 1$ and the cell of $\Pi_0$ contains $O(b')$ cells of $\Pi_1$.
  \item Any hyperplane in $H$ crosses at most $O((b'\cdot b^{i-1})^{1-1/d}+\log^{O(1)}m)$ cells of $\Pi_i$ for each $i\geq 1$.
\end{enumerate}
All collections of cells can be constructed in $O(m^d\cdot (r^{2-2/d}+\log^{O(1)}n)+n\log n+r^{3-1/d}+r^2\cdot \log^{O(1)}n+m\cdot r^{3-2/d})$ time.
\end{lemma}
\begin{proof}
We first assume that $r$ is a power of $b$. Then, $r=b^k$ and $b'=1$.

We apply Lemma~\ref{lem:partition} with $H$ to construct the sets $\Pi'_{t}$ iteratively for $t=1,b,b^2,\ldots$ until $b^k=r$.
More specifically, $\Pi'_1$ consists of a single cell that is $\bbR^d$, and $\Pi'_{bt}$ is obtained by applying Lemma~\ref{lem:partition} on all cells of $\Pi'_t$.
We then let $\Pi_0=\Pi_1=\Pi'_1$ and $\Pi_i=\Pi'_{b^{i-1}}$. By Lemma~\ref{lem:partition}, the properties (1)-(4) hold. In what follows, we prove Property (5), i.e., any hyperplane in $H$ crosses at most $O((b^i)^{1-1/d}+\log^{O(1)}m)$ cells of $\Pi'_{b^i}$.

Let $\kappa(t)$ denote the maximum number of cells of $\Pi'_t$ crossed by any hyperplane of $H$.  Since $\Pi'_1$ (which is $\Pi_0$) has a single cell, $\kappa(1)=1$. By \eqref{equ:bound}, we have
\begin{equation}\label{equ:30}
\begin{split}
\kappa(bt) &= O((bt)^{1-1/d}+b^{1-1/(d-1+\beta)}\cdot \kappa(t)\cdot \log t+b\cdot \log m)\\
&\leq c_0\cdot (bt)^{1-1/d}+c_0\cdot b^{1-1/(d-1+\beta)}\cdot \kappa(t)\cdot \log t+c_0\cdot b\cdot \log m
\end{split}
\end{equation}
for some constant $c_0$.

Fix $\delta>0$ to be an arbitrarily small constant. We prove the following by induction over $t=1,b,b^2,\ldots, b^k$.
\begin{align}\label{equ:40}
\kappa(t)\leq c \cdot (t^{1-1/d}+t^{1-1/(d-1+\beta)+\delta}\cdot b\cdot \log m)
\end{align}
for some sufficiently large constant $c$ depending on $\rho$ and $c_0$.
Clearly, $\kappa(1)=1$ satisfies~\eqref{equ:40} for any $c\geq 1$.
To prove \eqref{equ:40} for any $t$, in light of \eqref{equ:30}, it suffices to prove the following
\begin{equation}\label{equ:50}
\begin{split}
&c_0 \cdot (bt)^{1-1/d}+c_0\cdot b^{1-1/(d-1+\beta)}\cdot \log t\cdot \left[c \cdot (t^{1-1/d}+t^{1-1/(d-1+\beta)+\delta}\cdot b\cdot \log m)\right]+ c_0\cdot b\cdot \log m \\
&\leq c \cdot (bt)^{1-1/d}+c\cdot (bt)^{1-1/(d-1+\beta)+\delta}\cdot b\cdot \log m
\end{split}
\end{equation}
The LHS of \eqref{equ:50} is equal to
\begin{equation}\label{equ:60}
    \begin{split}
    &c_0\cdot (bt)^{1-1/d} + c_0\cdot c\cdot b^{1-1/(d-1+\beta)} \cdot t^{1-1/d} \cdot \log t \\
    &+ c_0\cdot c\cdot b^{1-1/(d-1+\beta)} \cdot t^{1-1/(d-1+\beta)+\delta}\cdot b\cdot \log t\cdot \log m + c_0\cdot b\cdot \log m\\
    &= c_0\cdot \left[1+\frac{c\cdot \log t}{b^{1/(d-1+\beta)-1/d}} \right]\cdot (bt)^{1-1/d}  + c_0\cdot \frac{c}{b^{\delta}}\cdot (bt)^{1-1/(d-1+\beta)+\delta} \cdot b\cdot \log t\cdot \log m + c_0\cdot b\cdot \log m\\
    \end{split}
\end{equation}
In the following, we show that we can find a sufficiently large constant $c$ such that the first term in the RHS of \eqref{equ:60} is smaller than the first term of the RHS of \eqref{equ:50} and the second (resp., third) term in the RHS of \eqref{equ:60} is smaller than half of the second term of the RHS of \eqref{equ:50}; this will prove \eqref{equ:50}.

\begin{enumerate}
    \item We first prove that we can find a sufficiently large constant $c$ such that the first term in the RHS of \eqref{equ:60} is smaller than or equal to the first term of the RHS of \eqref{equ:50}. It suffices to prove the following $$c_0\cdot \left[1+\frac{c\cdot \log t}{b^{1/(d-1+\beta)-1/d}}\right]\leq c.$$
    Indeed, recall that $b=\log^\rho r$ and $t\leq r$. Since $\beta$ is an arbitrarily small constant, $d-1+\beta<d$. Hence, we can set $\rho$ to a sufficiently large constant so that $\rho\cdot (1/(d-1+\beta)-1/d)\geq 2$. Consequently, we have
    $$c_0\cdot \left[1+\frac{c\cdot \log t}{b^{1/(d-1+\beta)-1/d}}\right]\leq c_0\cdot \left[1+\frac{c}{\log r}\right],$$ which is smaller than a sufficiently large constant $c$ since $r$ is larger than a sufficiently large constant $r_0$.
    \item
    We then prove that the second term in the RHS of \eqref{equ:60} is smaller than half of the second term of the RHS of \eqref{equ:50}. Indeed, it suffices to prove
    $$c_0\cdot \frac{1}{b^{\delta}} \cdot \log t\leq \frac{1}{2}.$$
    Since $b=\log^\rho r$ and $t\leq r$, we can make the above hold by choosing a sufficiently large constant $\rho$.
    \item
    We finally prove that the third term in the RHS of \eqref{equ:60} is smaller than half of the second term of the RHS of \eqref{equ:50}. Indeed, it suffices to prove
    $$c_0\leq \frac{c}{2}\cdot (bt)^{1-1/(d-1+\beta)+\delta},$$
    which clearly holds for a sufficiently large $c$.
\end{enumerate}

This proves \eqref{equ:50} and thus also \eqref{equ:40}. Note that the first term of \eqref{equ:40} dominates when $t$ exceeds $\log^{c_1} m$ for a sufficiently large constant $c_1$. Hence, we can write $\kappa(t)=O(t^{1-1/d}+\log^{O(1)}m)$. This proves the lemma statement (5) for the case where $r$ is a power of $b$.

If $r$ is not a power of $b$, then $b'=\lceil r/b^k\rceil\leq b$. We apply Lemma~\ref{lem:partition} with $H$ and $b'$ to the single cell of $\Pi_0$ to obtain $\Pi_1$. Then, we apply Lemma~\ref{lem:partition} with $H$ with $t=b',b'\cdot b,b'\cdot b^2,\ldots$ until $b'\cdot b^k$ to obtain $\Pi_2,\Pi_3,\cdots\Pi_k$.
For $t=b'$, by \eqref{equ:bound}, we have $\kappa(t) = O(t^{1-1/d}+t^{1-1/(d-1+\beta)} + b\log m)$, which satisfies \eqref{equ:40}. As such, following the same analysis leads to $\kappa(t)=O(t^{1-1/d}+\log^{O(1)}m)$.

This proves the properties (1)-(5). The algorithm implementation for constructing all collections of cells and the time analysis will be discussed in the next subsection.
\end{proof}

\subsection{The algorithm implementation for Lemma~\ref{lem:partitiontree}}

We now give an efficient implementation for the algorithm of Lemma~\ref{lem:partitiontree} and prove the time complexity as stated in Lemma~\ref{lem:partitiontree}. We again first assume that $r$ is a power of $b$, and thus $r=b^k$ and $b'=1$.

Recall that we have a set $S$ of $t$ cells for $t=1,b,b^2,\ldots,b^k$. For each $t$, we call it the {\em round-$t$} of the algorithm, which has $t$ iterations. In what follows, we describe the rount-$t$ of the algorihtm.

Consider the $i$-th iteration with $i$ running from $t$ down to $1$. In the beginning of the iteration, cells $\Delta_t,\Delta_{t-1},\ldots,\Delta_{i+1}$ have already been processed. Recall that $S_i$ is the subset of cells of $S$ that are not processed. In the $i$-th iteration, we need to first determine cell $\Delta_i\in S_i$ and then subdivide it into $O(b)$ subcells. Also recall that in the beginning the iteration we have a multiset $H_i$.
Let $\kappa$ be the maximum number of cells of $S$ crossed by any hyperplane of $H$. We have proved that $\kappa=O(t^{1-1/d}+\log^{O(1)}m)$. For notational convenience, we let $\kappa=\Theta(t^{1-1/d}+\log^{O(1)}m)$.

\subsubsection{Determining $\boldsymbol{\Delta_i}$}

We first need to determine $\Delta_i\in S_i$. This depends on the values $A_{\Delta}$ and $B_{\Delta}$ for all cells $\Delta \in S_i$, which in turn depend on the values $|H_i(\Delta)|$ and $N_{\Delta}(H_i)$. In what follows, we first discuss how to ``estimate'' $|H_i(\Delta)|$ and $N_{\Delta}(H_i)$, and then determine $\Delta_i$.

\paragraph{Estimating $\boldsymbol{|H_i(\Delta)|}$.}
For each each hyperplane $h\in H$, we attempt to maintain a weight $w(h)$ that is the multiplicity of $h$ in $H_i$. According to our algorithm, we have $w(h)=C\cdot (1+1/b)^{\lambda(h)}$, where $\lambda(h)$ is the cells crossed by $h$ (among the cells that have been produced since the beginning of the round-$t$ of the algorithm). Clearly, $\lambda(h)=O(bt)$. Our algorithm will maintain $\lambda(h)$ (i.e., we consider $\lambda$ as an array indexed by $h$) and we will discuss later about how to update $\lambda(h)$ once it changes. Define $I(h)=\lceil\log C+\lambda(h)\cdot \log(1+1/b)\rceil$ and $w'(h)=2^{I(h)}$. As such, $w(h)\in (w'(h)/2,w'(h)]$.

For each cell $\Delta\in S_i$, we use the following way to ``estimate'' $|H_i(\Delta)|$. Define $H'_i(\Delta)$ (resp., $H_i'$) as the set of distinct hyperplanes of $H_i(\Delta)$ (resp.,  $H_i$). Note that $|H_i(\Delta)|=\sum_{h\in H'_i(\Delta)}w(h)$. Define $W'(\Delta)=\sum_{h\in H'_i(\Delta)}w'(h)$. Since $w(h)\in (w'(h)/2,w'(h)]$, it holds that $|H_i(\Delta)|\in (W'(\Delta)/2,W'(\Delta)]$. We have the following observation on the value of $W'(\Delta)$.
\begin{observation}\label{obser:10}
$\log W'(\Delta)=O(\kappa)$, and thus $\log |H_i(\Delta)|=O(\kappa)$.
\end{observation}
\begin{proof}
Clearly, $|H_i(\Delta)|\leq |H_0|$. Recall that we have proved in Lemma~\ref{lem:partition} that $|H_0|$ is bounded by \eqref{equ:25}. Hence,
$$\log |H_i(\Delta)|=O\left(\log m+\frac{t^{1-1/d}}{b^{1/d}}+\frac{\kappa\cdot \log t}{b^{1/(d-1+\beta)}}\right).$$
Recall that $t\leq r$ and $b=\log^{\rho}r$ for a sufficiently large constant $\rho$. If we make $\rho$ large enough so that $\rho/(d-1+\beta)\geq 1$, then we have $\log t/b^{1/(d-1+\beta)}\leq 1$.
Since $\kappa=\Theta(t^{1-1/d}+\log^{O(1)}m)$,
we obtain that $\log |H_i(\Delta)|=O(\kappa)$. As $|H_i(\Delta)|\in (W'(\Delta)/2,W'(\Delta)]$, it follows that $\log W'(\Delta)=O(\kappa)$.
\end{proof}

In light of Observation~\ref{obser:10}, we use an array $E_{\Delta}$ of size $O(\kappa)$ to calculate the binary code of $W'(\Delta)$. Each array element of $E_{\Delta}$ is either zero or one. Initially every element of $E_{\Delta}$ is zero. For each hyperplane $h\in H'_i(\Delta)$, we want to add $w'(h)$ to $E_{\Delta}$ using binary operations.  Recall that $w'(h)=2^{I(h)}$ and $I(h)={\lceil\log C+\lambda(h)\cdot \log(1+1/b)\rceil}$. Thus, adding $w'(h)$ to $E_{\Delta}$ is equivalent to adding one to the $I(h)$-th element of $E_{\Delta}$. Assuming that $\lambda(h)$ is available, $I(h)$ can be calculated in $O(1)$ time. Adding one to the $I(h)$-th element of $E_{\Delta}$ can be easily done in $O(\kappa)$ time as the size of $E_{\Delta}$ is $O(\kappa)$, e.g., by scanning the array. Doing this for all hyperplanes $h\in H'_i(\Delta)$ will compute the binary code of $W'(\Delta)$; the time is bounded by $O(\kappa\cdot |H_i'(\Delta)|$.
As such, the total time for computing $W'(\Delta)$ for all cells $\Delta\in S_i$ is $\kappa\cdot \sum_{\Delta\in S_i}|H_i'(\Delta)|$. Because each hyperplane $h\in H$ crosses at most $\kappa$ cells of $S$, it holds that $\sum_{\Delta\in S_i}|H_i'(\Delta)|\leq \kappa\cdot m$. Therefore, the total time for computing $W'(\Delta)$ for all cells $\Delta\in S_i$ is $O(m\cdot \kappa^2)$.

\paragraph{Estimating $\boldsymbol{|N_{\Delta}(H_i)|}$.}
Let $\calA$ denote the arrangement of $H$.
For each vertex $v\in \calA$, let $w(v)$ denote the multiplicity of $v$, i.e., the multiplication of $w(h)$ of all hyperplanes $h$ through $v$. More specifically, if $H_v$ is the subset of hyperplanes of $H$ through $v$, then $w(v)=C^{|H_v|}\cdot (1+1/b)^{\sum_{h\in H_v}\lambda(h)}$.
Hence, $N_{\Delta}(H_i)=\sum_{v\in \calA(\Delta)}w(v)$, where $\calA(\Delta)$ denotes the portion of $\calA$ inside $\Delta$. Define $I(v)=\lceil\log (C^{|H_v|}\cdot (1+1/b)^{\sum_{h\in H_v}\lambda(h)})\rceil=\lceil |H_v|\cdot C+\sum_{h\in H_v}\lambda(h)\cdot \log(1+1/b)\rceil$. Our algorithm maintains the value $I(v)$ for each vertex $v\in \calA$. Note that once $\lambda(h)$ for a hyperplane $h\in H_v$ is changed, $I(v)$ can be updated in $O(1)$ time. Define $w'(v)=2^{I(v)}$. Hence, $w(v)\in (w'(v)/2,w'(v)]$.  Define $N'_{\Delta}(H_i)=\sum_{v\in \calA(\Delta)}w'(v)$. Clearly, $N_{\Delta}(H_i)\in (N'_{\Delta}(H_i)/2,N'_{\Delta}(H_i)]$. We will use $N'_{\Delta}(H_i)$ to ``estimate'' $N_{\Delta}(H_i)$.

Consider any cell $\Delta\in S_i$.
Recall that $\log |H_i(\Delta)|=O(\kappa)$ due to Observation~\ref{obser:10}. As $N_{\Delta}(H_i)\leq |H_i(\Delta)|^d$, we have $\log N_{\Delta}(H_i)\leq d\log |H_i(\Delta)|=O(\kappa)$. Hence, $\log N'_{\Delta}(H_i)=O(\kappa)$. We use an array $F_{\Delta}$ of size $O(\kappa)$ to calculate the binary code of $N'_{\Delta}(H_i)$. Our algorithm maintains $N'_{\Delta}(H_i)$ for all cells $\Delta\in S$. Whenever $\lambda(h)$ changes for some hyperplane $h$ defining a vertex $v\in \calA(\Delta)$, we first update $I(v)$ in $O(1)$ time and then update $N'_{\Delta}(H_i)$ in $O(\kappa)$ time using the array $F_{\Delta}$ (e.g., we first subtract one from the old $I(v)$-th element of $F_{\Delta}$ and then add one to the updated $I(v)$-th element; both operations can be done in $O(\kappa)$ time as the size of $F_{\Delta}$ is $O(\kappa)$).

\paragraph{Determining $\boldsymbol{\Delta_i}$.}
We now use the arrays $E_{\Delta}$ and $F_{\Delta}$ of $\Delta\in S_i$ to determine $\Delta_i$.
For each $\Delta\in S_i$, we need to first compare
$A_{\Delta}=\left(\frac{N_{\Delta}(H_i)}{b}\right)^{1/d}$ with $B_{\Delta}=\frac{|H_i(\Delta)|}{b^{1/(d-1+\beta)}}$. It is equivalent to compare $b^{(1-\beta)/(d-1+\beta)}\cdot N_{\Delta}(H_i)$ and $|H_i(\Delta)|^d$. Comparing them is challenging as we do not really know the exact values of $N_{\Delta}(H_i)$ and $|H_i(\Delta)|$. Our strategy is to compare $b^{(1-\beta)/(d-1+\beta)}\cdot N'_{\Delta}(H_i)$ and $|W'(\Delta)|^d$ instead. To this end, we use the largest index of $F_{\Delta}$ (resp., $E_{\Delta}$) whose element value is equal to one, denoted by $I_F(\Delta)$ (resp., $I_E(\Delta)$). Then, $I_F(\Delta)$ essentially represents the value $2^{I_F(\Delta)}$, which is within a factor at most $2$ of $N'_{\Delta}(H_i)$ and thus is within a factor at most $4$ of $N_{\Delta}(H_i)$. Similarly, $2^{I_E(\Delta)}$ is within a factor at most $2$ of $W'(\Delta)$ and thus is within a factor at most $4$ of $|H_i(\Delta)|$. Hence, $2^{d\cdot I_E(\Delta)}$ is within a factor at most $4^d$ of $|H_i(\Delta)|^d$.
We compare $b^{(1-\beta)/(d-1+\beta)}\cdot 2^{I_F(\Delta)}$ with $2^{d\cdot I_E(\Delta)}$, which is equivalently to comparing $(1-\beta)/(d-1+\beta)\log b+I_F(\Delta)$ with $d\cdot I_E(\Delta)$. The latter comparison can be done in $O(1)$ time once we have $I_E(\Delta)$ and $I_F(\Delta)$, which can be found from $E_{\Delta}$ and $F_{\Delta}$ in $O(\kappa)$ time, respectively.
We compute $I_E(\Delta)$ and $I_F(\Delta)$ for all cells $\Delta\in S_i$, which takes $O(\kappa\cdot t)$ time.

Instead of using $S_{i1}$ and $S_{i2}$ as defined before, we define two different sets $S'_{i1}$ and $S'_{i2}$. Let $S'_{i1}$ consist of cells $\Delta\in S_i$ with $(1-\beta)/(d-1+\beta)\log b+I_F(\Delta)\geq d\cdot I_E(\Delta)$; let $S'_{i2}=S\setminus S'_{i1}$. We can compute these two sets in $O(t)$ time. We instead use $S'_{i1}$ and $S'_{i2}$ and follow the same algorithm as before, i.e., replacing $S_{i1}$ and $S_{i2}$ in our original algorithm in the proof of Lemma~\ref{lem:partition} by $S'_{i1}$ and $S'_{i2}$, respectively. Specifically, if $|S'_{i1}|\geq i/2$, we let $\Delta_i$ be the cell $\Delta\in S'_{i1}$ with smallest $I_F(\Delta)$; otherwise let $\Delta_i$ be the cell $\Delta\in S'_{i1}$ with smallest $I_E(\Delta)$. Hence, $\Delta_i$ can be determined in $O(t)$ time. If we follow the same algorithm using this new defined $\Delta_i$, then using the same analysis, we can still obtain the same bound as before in the proof of Lemma~\ref{lem:partition} except that the bound for each $\alpha_i$ is a constant factor difference from before. Indeed, this is because $2^{I_F(\Delta)}$ is within a factor at most $4$ of $N_{\Delta}(H_i)$, $2^{d\cdot I_E(\Delta)}$ is within a factor at most $4^d$ of $|H_i(\Delta)|$, and $d$ is a constant. In particular, we can still get the bound~\eqref{equ:bound}.

In summary, determining $\Delta_i$ can be done in $O(m\kappa^2+\kappa\cdot t)$ time. Doing this for all $t$ iterations takes $O(mt\kappa^2+\kappa\cdot t^2)$ time.

\subsubsection{Processing $\boldsymbol{\Delta_i}$}
Recall that we have three steps to process $\Delta_i$. We discuss how to implement them below.

\paragraph{The first step.}
The first step is to compute a $(1/r_i)$-cutting for $H_i$ inside $\Delta_i$. The definition of $r_i$ depends on $|H_i(\Delta_i)|$ and $N_{\Delta_i}(H_i)$. Instead, we also use $2^{I_E(\Delta)}$ and $2^{I_F(\Delta)}$. Specifically, to calculate $|H_i(\Delta_i)|\cdot (\frac{b}{N_{\Delta_i}(H_i)})^{1/d}$, we attempt to calculate $2^{I_E(\Delta)}\cdot (\frac{b}{2^{I_F(\Delta)}})^{1/d}=2^{I_E(\Delta)-I_F(\Delta)/d}\cdot b^{1/d}$. Instead of calculating the value $2^{I_E(\Delta)-I_F(\Delta)/d}\cdot b^{1/d}$ directly (which may be too large), we first compare $2^{I_E(\Delta)-I_F(\Delta)/d}\cdot b^{1/d}$ with $b^{1/(d-1+\beta)}$, which is equivalent to comparing $I_E(\Delta)-I_F(\Delta)/d$ with $(1/(d-1+\beta)-1/d)\cdot \log b$; the latter comparison can be done in $O(1)$ time. If $I_E(\Delta)-I_F(\Delta)/d>(1/(d-1+\beta)-1/d)\cdot \log b$, then we set $r_i=c\cdot b^{1/(d-1+\beta)}$; otherwise, $2^{I_E(\Delta)-I_F(\Delta)/d}\cdot b^{1/d}\leq b^{1/(d-1+\beta)}$ and thus $2^{I_E(\Delta)-I_F(\Delta)/d}$ is small enough and can be computed in $O(\kappa)$ time as both $I_E(\Delta)$ and $I_F(\Delta)$ are $O(\kappa)$. In the latter case, we set $r_i=c\cdot 2^{I_E(\Delta)-I_F(\Delta)/d}\cdot b^{1/d}$. Note that the small constant $c$ can be determined from the analysis. As such, computing this new $r_i$ can be done in $O(\kappa)$ time.
The new $r_i$ is only a constant difference from the originally defined $r_i$ in the proof of Lemma~\ref{lem:partition}, and thus this does not affect the analysis asymptotically.

We next compute a $(1/r_i)$-cutting for $H_i(\Delta_i)$. This is equivalent to computing a $(1/r_i)$-cutting for $H'_i(\Delta_i)$ where each hyperplane $h$ has a weight $w(h)$ as defined above (in the weighted case, the total weight of all hyperplanes crossing any cell of the cutting is required to be no more than the total weight of all hyperplanes of $H_i'(\Delta_i)$ divided by $r_i$). The challenge is that we do not know the value $w(h)$ exactly. Instead, we use $w'(h)$ as the weight. Since $w(h)\in (w'(h)/2,w'(h)]$, the sum of the weights of all hyperplanes in $H_i(\Delta_i)$ is at most a factor of 2 difference from the sum of the old weights using $w(h)$. Therefore, a $(1/(4r_i))$-cutting using the new weights $w'(h)$ must be a $(1/r_i)$-cutting on the old weights $w(h)$. To compute such a cutting, even though we know the values $w'(h)$ exactly, we cannot apply the cutting algorithm~\cite{ref:ChazelleCu93} directly because each $w'(h)$ may be exponentially large. To cope with this issue, we have the following lemma.

\begin{lemma}\label{lem:weightedcutting}
A $(1/(4r_i))$-cutting for $H'_i(\Delta_i)$ using the weights $w'(h)$
can be computed in $O(|H'_i(\Delta_i)|\cdot (\kappa+r_i^{d-1}))$ time.
\end{lemma}
\begin{proof}
Let $Q=H_i'(\Delta_i)$ and $m_Q=|Q|$.
Recall that for each hyperplane $h\in Q$, $\lambda(h)$ is available to us. Recall also that $w'(h)=2^{I(h)}$ with $I(h)=\lceil\log C+\lambda(h)\cdot \log(1+1/b)\rceil$. Define $W'(Q)=\sum_{h\in Q}w'(h)$, i.e., $W'(Q)=W'(\Delta_i)$. By Observation~\ref{obser:10}, $\log W'(Q)=O(\kappa)$.


We use an approach similar to that in \cite{ref:WangAl20}, which extends the method suggested by Matou\v{s}ek (in Lemma 3.4~\cite{ref:MatousekEf92}) and the algorithm in Theorem 2.8
of~\cite{ref:MatousekCu91} for computing a cutting for a set
of weighted lines.

The first step is to compute an integer $q$ so that $2^q\leq W'(Q)<2^{q+1}$.
For each hyperplane $h\in Q$, since $w'(h)\leq W'(Q)$, we also have $I(h)=O(\kappa)$. We can compute $q$ in $O(m_{\Delta}\cdot\kappa)$ time as follows.

Create an array $A$ of size $O(\kappa)$, with each element of $A$ equal to either $1$ or $0$.
For analysis purpose, let $val(A)$ denote the value represented by the binary code of
the elements of $A$. Initially, we set every element of $A$ to $0$, and thus $val(A)=0$.
For each hyperplane $h\in Q$, we add $2^{I(h)}$ to $val(A)$ by updating the
array $A$, i.e., add $1$ to the $I(h)$-th element of $A$.
Since the size of $A$ is $O(\kappa)$, this addition operation
can be easily done in $O(\kappa)$ time. The total time for
doing this for all hyperplanes of $Q$ is $O(m_{\Delta}\cdot\kappa)$ time. Finally, if $i$ is the
largest index of $A$ with $A[i]=1$, then we have $q=i$.

Define $p=\lfloor\log m_{Q}\rfloor$. Thus, $2^p\leq m_{Q}< 2^{p+1}$.

We define a multiset $Q'$ of $Q$ as follows. For each hyperplane $h\in Q$, if
$p+1+I(h)-q\geq 0$, then we place $2^{p+1+I(h)-q}$ copies of $h$ into
$Q'$; otherwise, we put just one copy of $h$ into $Q'$. Let $|Q'|$
denote the cardinality of $Q'$, counted with multiplicities. Then, we can derive
\begin{equation*}
\begin{split}
|Q'| & \leq   |Q| + \sum_{h\in Q} 2^{p+1+I(h)-q} = m_Q +
2^{p+1-q} \cdot  \sum_{h\in Q} 2^{I(h)} = m_Q + 2^{p+1-q} \cdot W'(Q)
\\
& \leq  m_Q + 2^{p+1-q} \cdot 2^{q+1} = m_Q + 2^{p+2}\leq m_Q + 4m_Q = 5m_Q.
\end{split}
\end{equation*}

This bound $|Q'|\leq 5m_Q$ also implies that the step of ``place $2^{p+1+I(h)-q}$ copies of $h$ into $Q'$'' for all $h\in Q$ can be done in $O(m_{Q})$ time. As such, producing the multiset $Q'$ takes $O(m_{Q})$ time.

Let $r=r_i$. We now compute a $\frac{1}{5r}$-cutting $\Xi$ for the
unweighted multiset $Q'$ in $O(m_{Q}\cdot r^{d-1})$ time by Lemma~\ref{lem:cutting}. We prove below that $\Xi$ is a $(1/r)$-cutting for the weighted set $(Q,w')$, and therefore
we can simply return $\Xi$ as our solution to the lemma. The total time of the algorithm
is $O((m_{Q}\cdot(\kappa+r^{d-1}))$.

Let $\sigma$ be a cell of $\Xi$. Define $Q_{\sigma}$ to
be the subset of hyperplanes of $Q$ that cross $\sigma$. It suffices to prove $W'(Q_{\sigma})\leq W'(Q)/r$, where $W'(Q_{\sigma})$ is the total weight of all hyperplanes of $Q_{\sigma}$.
Define $Q'_{\sigma}$ as the multiset of hyperplanes of $Q'$ crossing
$\sigma$.
Since $\Xi$ is a $\frac{1}{5r}$-cutting of $Q'$ and $|Q'|\leq 5m_Q$, we have $|Q'_{\sigma}|\leq \frac{|Q'|}{5r}\leq m_Q/r$.
Consequently, we can derive
\begin{equation*}
\begin{split}
W'(Q_{\sigma}) & = \sum_{h\in Q_{\sigma}} w'(h) = \sum_{h\in
Q_{\sigma}} 2^{I(h)} = \frac{1}{2^{p+1-q}} \cdot \sum_{h\in
Q_{\sigma}} 2^{p+1+I(h)-q} \leq  \frac{1}{2^{p+1-q}} \cdot
|Q'_{\sigma}| \\
& \leq \frac{m_Q}{2^{p+1-q}\cdot r} = \frac{2^q\cdot m_Q}{2^{p+1}\cdot r}
\leq \frac{W'(Q)\cdot m_Q}{2^{p+1}\cdot r} \leq \frac{W'(Q)}{r}.\\
\end{split}
\end{equation*}

This proves that $\Xi$ is a $\frac{1}{r}$-cutting for $(Q,w')$.
\end{proof}

Since $r_i=O(b^{1/(d-1+\beta)})$, the time of the algorithm in the above lemma is bounded by $O(|H'_i(\Delta_i)|\cdot (\kappa+b))$.
The total time of computing the cuttings for all $t$ iterations is $O(\sum_{\Delta\in S}|H'_i(\Delta)|\cdot (\kappa+b))$. Since each hyperplane of $H$ crosses $O(\kappa)$ cells of $S$, we have $\sum_{\Delta\in S}|H'_i(\Delta_i)|=O(m\kappa)$, and thus the total time of the cutting algorithm for all $t$ iterations is $O(m\kappa(\kappa+b))$.

\paragraph{The second step.}
The second step of the algorithm is to make vertical cuts so that each subcell of the cutting has at most $n/(bt)$ points of $P$. For this, we presort all points of $P$ in the beginning of the entire algorithm. Assuming that the at most $n/t$ points of $P$ in $\Delta_i$ have been sorted. Then, for each point in the sorted order of the points of $P$ inside $\Delta_i$, we first determine the subcell of the cutting containing the point, which takes $O(\log r_i)=O(\log b)$ time by point location~\cite{ref:ChazelleCu93}. In this way, we obtain the sorted list of the points of $P$ in each subcell of the cutting. The vertical cuts can be made by traversing the sorted list in linear time. Hence this step can be done in $O(n/t\cdot \log b)$ time. The total time of this step for all $t$ iterations is $O(n\log b)$, excluding the presorting step that takes $O(n\log n)$ time.

\paragraph{The third step.}
In the third step, we need to count for each hyperplane $h\in H_i'(\Delta_i)$ the number of subcells of $\Delta_i$ crossed by $h$. To this end, we use a simple brute-force algorithm. For each hyperplane $h\in H'_i(\Delta_i)$, for each subcell of $\Delta_i$, we check whether $h$ crosses the subcell. As $\Delta_i$ has $O(b)$ cells, the time of this step is bounded by $O(|H'_i(\Delta_i)|\cdot b)$.
The total time of this step for all $t$ iterations is $\sum_{\Delta_i\in S}O(|H'_i(\Delta_i)|\cdot b)$, which is $O(m\cdot \kappa\cdot b)$ as $\sum_{\Delta_i\in S}|H'_i(\Delta_i)|=O(m\kappa)$.

After this step, we update $\lambda(h)$ for each $h\in H'_i(\Delta_i)$. We also need to update the arrays $E_{\Delta}$ and $F_{\Delta}$ for cells $\Delta\in S_{i-1}=S_{i}\setminus\Delta_i$. To facilitate the update, before the first iteration of the algorithm (i.e., before $\Delta_t$ is processed), we construct the arrangement $\calA$ of $H$ and for each hyperplane $h\in H$, we maintain a list of cells of $S$ crossed by $h$ and a list of vertices $v$ of $\calA$ that are on $h$ as well as the cells of $S$ containing $v$. In this way, we can use the above two lists associated with $h$ to update the corresponding $E_{\Delta}$ and $F_{\Delta}$. More specifically, for each $h\in H'_i(\Delta_i)$, if $\lambda(h)$ is updated, then using the first list, we find those cells  $\Delta\in S_{i-1}$ crossed by $h$ and update $E_{\Delta}$ in $O(\kappa)$ time as discussed before; using the second list, we find the vertices $v$ of $\calA$ on $h$, and for each such $v$, we update $I(v)$ in $O(1)$ time and also update $F_{\Delta}$ in $O(\kappa)$ time, where $\Delta$ is the cell of $S$ containing $v$. Since each hyperplane of $H$ crosses $O(\kappa)$ cells of $S$, the total time of these updates in all $t$ iterations is $O(m^d\cdot \kappa^2)$.

Therefore, processing $\Delta_i$ for all $t$ iterations takes $O(m^d\cdot \kappa^2+m\kappa(\kappa+b)+n\log b)$ time, excluding the time for sorting $P$.

\subsubsection{Summary}

In summary, for each $t=1,b,b^2,\cdots,b^k$, the runtime of the round-$t$ of the algorithm is bounded by $O(m^d\cdot \kappa^2+m\kappa(\kappa+b)+mt\kappa^2+\kappa\cdot t^2+n\log b)$, with $\kappa=O(t^{1-1/d}+\log^{O(1)}m)$ and $b=\log^{\rho} r$, excluding the time for sorting $P$.
Since $r\leq n$ and $m=n^{O(1)}$, the total time of the overall algorithm for all $t=1,b,b^2,\cdots,b^k$ is $O(m^d\cdot (r^{2-2/d}+\log^{O(1)}n)+n\log n+r^{3-1/d}+r^2\cdot \log^{O(1)}n+m\cdot r^{3-2/d}+m\cdot r^{1-1/d}\cdot \log^{O(1)}n)$. Note that $m\cdot r^{1-1/d}\cdot \log^{O(1)}n=O(m^d\log^{O(1)}n+r^2\cdot \log^{O(1)}n)$. Indeed, if $m\leq r$, then $m\cdot r^{1-1/d}\cdot \log^{O(1)}n=O(r^2\cdot \log^{O(1)}n)$; otherwise, $m\cdot r^{1-1/d}\cdot \log^{O(1)}n=O(m^d\log^{O(1)}n)$ since $d\geq 2$.
Therefore, the runtime is bounded by $O(m^d\cdot (r^{2-2/d}+\log^{O(1)}n)+n\log n+r^{3-1/d}+r^2\cdot \log^{O(1)}n+m\cdot r^{3-2/d})$, as stated in Lemma~\ref{lem:partitiontree}.

If $b'\neq 1$, then the only difference is that in the first round we use parameter $b'\leq b$. The algorithm and analysis are similar; the runtime is asymptotically the same as above.
This proves Lemma~\ref{lem:partitiontree}.


\subsection{The hierarchical partition tree}

Lemma~\ref{lem:partitiontree} guarantees the crossing numbers for the hyperplanes in $H$. To make it work for any hyperplane in $\bbR^d$, we use the following lemma, which was known before, e.g.,~\cite{ref:ChanOp12}.


\begin{lemma}\label{lem:test20}{\em(Test Set Lemma~\cite{ref:ChanOp12})}
For a set $P$ of $n$ points in $\bbR^d$, we can compute in $O(n^d)$ time a set $H$ (called {\em test set}) of $O(n^d)$ hyperplanes with the following property: for any set of cells each containing at least one point of $P$, if $\kappa$ is the maximum number of cells crossed by any hyperplane of $H$, then the maximum number of cells crossed by any hyperplane is $O(\kappa)$.
\end{lemma}
\begin{proof}
As discussed in \cite{ref:ChanOp12}, we can take all $O(n^d)$ hyperplanes that contain all $d$-tuple of points of $P$, which has the desired property of the lemma. These hyperplanes are dual to the vertices of the arrangement of all dual hyperplanes of the points of $P$. As the arrangement can be computed in $O(n^d)$ time~\cite{ref:EdelsbrunnerCo86}, the lemma follows. \end{proof}

Lemmas \ref{lem:partitiontree} and \ref{lem:test20} together lead to the following result.

\begin{theorem}\label{theo:partition}
{\em (Hierarchical Partition Theorem)}
Given a set $P$ of $n$ points in $\bbR^d$ and a parameter $r$ with $r_0\leq r\leq n$ for a sufficiently large constant $r_0$, for $b=\log^\rho r$ for a sufficiently large constant $\rho>0$, there exists a sequence  $\Pi_0,\Pi_1,\ldots,\Pi_{k+1}$ with $k=\lfloor \log_b r\rfloor$ and $b'=\lceil r/b^{k}\rceil$, where each $\Pi_i$ is a collection of disjoint simplicial cells (i.e., each cell is a simplex) whose union covers $P$ with the following properties:
\begin{enumerate}
  \item $\Pi_0$ has only one cell that is $\bbR^d$ and the number of cells of $\Pi_i$ is $O(b'\cdot b^{i-1})$ for $i\geq 1$ (in particular, $\Pi_{k+1}$ has $O(r)$ cells; the total number of cells in all collections is also $O(r)$).
  \item Each cell of $\Pi_i$ contains at least one point and at most $2n/(b'\cdot b^{i-1})$ points of $P$ for all $i\geq 0$ (in particular, each cell of $\Pi_{k+1}$ contains at most $2n/r$ points).
  \item Each cell in $\Pi_{i+1}$ is contained in a single cell of $\Pi_i$.
  \item Each cell of $\Pi_i$ contains $O(b)$ cells of $\Pi_{i+1}$ for $i\geq 1$ and the cell of $\Pi_0$ contains $O(b')$ cells of $\Pi_1$.
  \item Any hyperplane crosses at most $O((b'\cdot b^{i-1})^{1-1/d}+\log^{O(1)}n)$ cells of $\Pi_i$ for all $i\geq 1$.
\end{enumerate}
All collections of cells can be constructed in $O(n^{d^2}\cdot (r^{2-2/d}+\log^{O(1)}n))$ time.
\end{theorem}
\begin{proof}
We just apply Lemma~\ref{lem:partitiontree} with $H$ as the test set given in Lemma~\ref{lem:test20}, which can be computed in $O(n^d)$ time. Since $m=O(n^d)$ and $r\leq n$, we obtain the theorem from Lemma~\ref{lem:partitiontree}.
\end{proof}

The collections of cells of Theorem~\ref{theo:partition} naturally form a tree structure, called a {\em hierarchical partition tree}, in which each node corresponds to a cell. More specifically, the only cell in $\Pi_0$ is the root. Cells of $\Pi_{k+1}$ are the leaves. If a cell $\Delta\in \Pi_i$ contains another cell $\Delta'\in \Pi_{i+1}$, then $\Delta$ is the parent of $\Delta'$ and $\Delta'$ is a child of $\Delta$. As such, each node of the tree has $O(b)$ children (the root has $O(b')$ children).
Although the construction time of Theorem~\ref{theo:partition} is relatively large, it can usually be reduced (e.g., to $O(n^{1+\epsilon})$ time) in applications by constructing an ``upper'' partition tree using other techniques (e.g., \cite{ref:MatousekEf92}) and then processing the leaves (each containing a small number of points) using Theorem~\ref{theo:partition}, as will be demonstrated in the subsequent sections.


\section{Simplex range counting}
\label{sec:simcount}

Given a set $P$ of $n$ points in $\bbR^d$, the problem is to construct a data structure so that the number of points of $P$ inside a query simplex can be quickly computed.
If we set $r=n$ in Theorem~\ref{theo:partition}, we can have a data structure of $O(n)$ space and $O(b\cdot n^{1-1/d})$ query time, which is not $O(n^{1-1/d})$ as $b=\Theta(\log n)$.
In the following, we reduce the query time to $O(n^{1-1/d}/\log^{\Omega(1)} n)$.

For any region $R$ in the plane, let $P(R)$ denote the subset of points of $P$ in $R$.

Setting $r=n/b^{2d/(d-1)}$ with $b=\log^\rho n$, we apply Theorem~\ref{theo:partition} to obtain a partition tree $T$ with collections $\Pi_i$, $0\leq i\leq k+1$. The total number of cells is $O(r)$. For each cell $\Delta$, we store $|P(\Delta)|$. By Theorem~\ref{theo:partition}, each cell in $\Pi_{k+1}$ contains $O(n/r)=O(b^{2d/(d-1)})=O(\log^{2d\rho/(d-1)} n)$ points of $P$. The runtime to construct $T$ is $O(n^{d^2+2})$.

Given a query simplex $\sigma$, starting from the root of $T$, for each cell $\Delta$ of $T$, if $\Delta$ is contained in $\sigma$, then we add $|P(\Delta)|$ to a total count. If $\Delta$ is completely outside $\sigma$, then we ignore it. Otherwise, a bounding halfplane of $\sigma$ must cross $\Delta$ and we proceed on all children of $\Delta$ unless $\Delta$ is a leaf. In this way, we obtain a set $V$ of $O(r^{1-1/d})$ leaf cells that are crossed by the bounding halfplanes of $\sigma$. The runtime is $O((r/b)^{1-1/d}\cdot  b)$, which is $O(r^{1-1/d}\cdot  b^{1/d})$.

We now build the data structure recursively on the points in $P(\Delta)$ for each leaf cell $\Delta$ of $T$.
If $Q(n)$ is the query time, we have the following recurrence relation:
\begin{equation}\label{equ:new10}
    Q(n)  = O(r^{1-1/d}\cdot b^{1/d}) + O(r^{1-1/d})\cdot Q_1(n/r),
\end{equation}
where $Q_1(\cdot)$ is the query time for each leaf cell $\Delta$, which contains $O(n/r)=O(\log^{2d\rho/(d-1)} n)$ points of $P$.

To solve $Q_1(n/r)$, we preprocess $P(\Delta)$ for each leaf cell $\Delta\in T$ recursively as above by using different parameters. Specifically, let $n_1=|P(\Delta)|$, which is $O(n/r)$. Let $\tau>0$ be an arbitrarily small constant to be set later. Setting $r_1=n_1/\log^{\tau}n$ with $b_1=\log^{\rho}n_1$, we apply Theorem~\ref{theo:partition} to construct a partition tree $T(\Delta)$ in $O(n_1^{d^2+2})$ time. The total time for constructing $T(\Delta)$ for all leaf cells $\Delta$ of $T$ is bounded by $O(n^{d^2+2})$.  Using $T(\Delta)$ to handle queries on $P(\Delta)$ and following the same analysis as above, we obtain
\begin{equation}\label{equ:new20}
    Q_1(n_1) = O(r_1^{1-1/d}\cdot b_1^{1/d}) + O(r_1^{1-1/d})\cdot Q_2(n_1/r_1),
\end{equation}
where $Q_2(\cdot)$ is the query time for each leaf cell of $T(\Delta)$, which contains $O(n_1/r_1)$ points of $P$.

Combining \eqref{equ:new10} and \eqref{equ:new20} leads to
\begin{equation*}
    \begin{split}
        Q(n) & = O(r^{1-1/d}\cdot b^{1/d}) + O(r^{1-1/d})\cdot Q_1(n/r)\\
             & = O(r^{1-1/d}\cdot b^{1/d}) + O(r^{1-1/d}\cdot r_1^{1-1/d}\cdot b_1^{1/d}) + O(r^{1-1/d}\cdot r_1^{1-1/d})\cdot Q_2(n_1/r_1)\\
             & = O(r^{1-1/d}\cdot b^{1/d}) + O\left(r^{1-1/d}\cdot \left(\frac{n_1}{\log^{\tau} n}\right)^{1-1/d}\cdot b_1^{1/d}\right) + O\left(r^{1-1/d}\cdot \left(\frac{n_1}{\log^{\tau} n}\right)^{1-1/d}\right)\cdot Q_2(n_1/r_1)\\
             & = O(r^{1-1/d}\cdot b^{1/d}) + O\left(r^{1-1/d}\cdot \left(\frac{n}{r\cdot \log^{\tau} n}\right)^{1-1/d}\cdot b_1^{1/d}\right) + O\left(r^{1-1/d}\cdot \left(\frac{n}{r\cdot \log^{\tau} n}\right)^{1-1/d}\right)\cdot Q_2(n_1/r_1)\\
             & = O(r^{1-1/d}\cdot b^{1/d}) + O\left(\left(\frac{n}{\log^{\tau} n}\right)^{1-1/d}\cdot b_1^{1/d}\right) + O\left(\left(\frac{n}{\log^{\tau} n}\right)^{1-1/d}\right)\cdot Q_2(n_1/r_1)\\
             & =  O\left(\frac{n^{1-1/d}}{\log^{(2-1/d)\cdot \rho} n}\right) + O\left(\left(\frac{n}{\log^{\tau} n}\right)^{1-1/d}\cdot \log^{\rho/d}\log n\right) + O\left(\left(\frac{n}{\log^{\tau} n}\right)^{1-1/d}\right)\cdot Q_2(\log^{\tau}n)\\
             & =  O\left(\frac{n^{1-1/d}}{\log^{\Omega(1)} n}\right) + O\left(\left(\frac{n}{t}\right)^{1-1/d}\right)\cdot Q_2(t),\\
    \end{split}
\end{equation*}
where $t=\log^{\tau}n$.


In summary, the above first builds a partition tree $T$ and then builds a partition tree $T(\Delta)$ for each leaf $\Delta$ of $T$ (so there are two recursive steps but with different parameters). For notational convenience, we still use $T$ to refer to the entire tree (by attaching $T(\Delta)$ for all leaves $\Delta$), which has $O(n/t)$ leaves, each containing $O(t)$ points. The total space is bounded by $O(n)$ since there are only two recursive steps. In Section~\ref{sec:simcountsub}, by using the property that $t$ is very small, we show that after $O(n)$ space and $O(n\log n)$ time preprocessing, each simplex range counting query on any leaf cell of $T$ can be answered in $O(\log t)$ time (i.e., $Q_2(t)=O(\log t)$, which is $O(\log\log n)$; we consider the query on each leaf cell a {\em subproblem}). As such, we obtain that $Q(n)=O(n^{1-1/d}/\log^{\Omega(1)} n)$. The total preprocessing time is $O(n^{d^2+2})$, dominated by the time for constructing $T$. We thus have the following result.

\begin{lemma}\label{lem:counting}
Given a set $P$ of $n$ points in $\bbR^d$, there is a data structure of $O(n)$ space that can compute the number of points of $P$ in any query simplex in $O(n^{1-1/d}/\log^{\Omega(1)} n)$ time. The data structure can be built in $O(n^{d^2+2})$ time.
\end{lemma}

We will further reduce the preprocessing time to $O(n^{1+\epsilon})$ in Section~\ref{sec:simcountpretime}.

\subsection{Solving the subproblems}
\label{sec:simcountsub}

We first consider halfspace range counting queries and then extend the technique to the simplex case.

\paragraph{A basic data structure.}
Let $A$ be a set of $t$ points in $\bbR^d$. We first build a straightforward data structure (called a {\em basic data structure}) of $O(t^d)$ space in $O(t^{d+1})$ time that can answer each halfspace range counting query on $A$ in $O(\log t)$ time.

Let $H$ be the set of dual hyperplanes of $A$. We compute the arrangement $\calA$ of $H$ in $O(t^d)$ time and space~\cite{ref:EdelsbrunnerCo86}. We then build a point location data structure on $\calA$, which can be done in $O(t^d)$ time and supports $O(\log t)$-time point location queries~\cite{ref:ChazelleCu93}. In addition, for each face $f$ of $\calA$, we compute the number of hyperplanes above $f$ (resp., below $f$) and store these two numbers at $f$, e.g., by checking every hyperplane of $A$. This finishes our preprocessing, which can be done in $O(t^{d+1})$ time and uses $O(t^d)$ space. The preprocessing time can be reduced to $O(t^d)$, e.g., by taking an Eulerian tour of the dual graph of the arrangement.

Given a query halfspace $\sigma$, the goal is to compute the number of points of $A$ inside $\sigma$. Without loss of generality, we assume that $\sigma$ is an upper halfspace.
Then, it is equivalent to computing the number of hyperplanes of $H$ below $p$, where $p$ is the dual point of the bounding hyperplane of $\sigma$. Using the point location data structure, we find the face $f$ of $\calA$ that contains $p$; $f$ stores the number of hyperplanes of $H$ below it and we return that number as our answer to the query. The query time is thus $O(\log t)$.

\paragraph{Handling halfspace queries.}
Next, we show that after $O(n\log n)$ time and $O(n)$ space preprorcessing we can answer each halfspace range counting query  in $O(\log t)$ time on $P(\Delta)$ for each leaf cell $\Delta$ of our partition tree $T$.

We build an algebraic decision tree $T_D$ for the arrangement construction algorithm~\cite{ref:EdelsbrunnerCo86} on a set of $t$ hyperplanes in $\bbR^d$ so that each node of $T_D$ corresponds to a comparison in the algorithm. The height of $T_D$ is $O(t^d)$ and $T_D$ has $2^{O(t^d)}$ leaves. Each leaf $v$ of $T_D$ corresponds to a ``configuration'' of a set of $t$ hyperplanes in the following sense. Let $A$ be a set of $t$ hyperplanes with indices from $1$ to $t$. Following the decision tree $T_D$ in a top-down manner, we can reach a leaf $v$ such that all comparisons of the nodes in the path of $T_D$ from the root to $v$ are consistent with $A$; we say that $A$ has the same configuration as $v$. Let $A$ and $B$ are two sets of $t$ hyperplanes each. Let $\calA_A$ and $\calA_B$ be the arrangements of $A$ and $B$, respectively.
If the configurations of $A$ and $B$ are both the same as a leaf $v$ of $T_D$, then
there is a one-to-one correspondence between faces of $\calA_A$ and faces of $\calA_B$ such that if a face $f$ of $\calA_A$ corresponds to a face $f'$ of $\calA_B$, then the $i$-th hyperplane of $A$ is above $f$ if and only if the $i$-th hyperplane of $B$ is above $f'$.

In light of the above observation, we do the following preprocessing. For each leaf $v$ of $T_D$, let $A_v$ be a set of $t$ hyperplanes whose configuration corresponds to that of $v$ (note that the sets $A_v$'s for all leaves $v$ can be constructed in $t^d\cdot 2^{O(t^d)}$ time by following $T_D$, which basically enumuerates all possible configurations for the arrangements of a set of $t$ hyperplanes). We construct the above basic data structure on $A_v$, denoted by $\calD_v$. Doing this for all leaves of $T_D$ takes $t^d\cdot 2^{O(t^d)}$ time and space, which is bounded by $O(n)$ since $t=\log^{\tau}n$ if we choose a small enough $\tau$.

For each leaf cell $\Delta$ of our partition tree $T$, recall that $|P(\Delta)|\leq t$ (if $|P(\Delta)|< t$, we can add $t-|P(\Delta)|$ dummy points to $P(\Delta)$ so that $P(\Delta)$ has exactly $t$ points). Let $H(\Delta)$ denote the set of dual hyperplanes of the points of $P(\Delta)$. We arbitrarily assign indices to the hyperplanes of $H(\Delta)$. Following the decision tree $T_D$, we find the leaf $v$ of $T_D$ that has the same configuration as $H(\Delta)$, which can be done in time linear in the height of $T_D$, i.e., $O(t^d)$; we associate $v$ with $\Delta$. Since $T$ has $O(n/t)$ leaves, doing this for all leaves of $T$ takes $O(n/t\cdot t^{d})$ time, which is $O(n\log n)$ if $\tau$ is small enough.

Given a query halfspace $\sigma$, suppose we want to compute the number of points of $P(\Delta)$ inside $\sigma$ for a leaf cell $\Delta$ of $T$. This can be done in $O(\log t)$ time as follows.
Let $v$ be the leaf of $T_D$ associated with $\Delta$. Let $p$ be the dual point of the bounding hyperplane of $\sigma$. Without loss of generality, we assume that $\sigma$ is an upper halfspace. Hence it is equivalent to finding the number of hyperplanes of $H(\Delta)$ below $p$. We apply the point location query algorithm using the data structure $\calD_v$ with $p$, but whenever the algorithm attempts to use the $i$-th hyperplane of $A_v$ to make a comparison, we use the $i$-th hyperplane of $H(\Delta)$ instead. The point location algorithm will eventually return a face, which stores the number of hyperplanes below it; we return that number as our answer to the query of $\sigma$.

\paragraph{Handling simplex queries.}
We can easily extend the above idea to simplex queries, by using a multilevel data structure. Specifically, for each leaf $v$ of $T_D$, we perform the following preprocessing on $A_v$. We call the data structure $\calD_v$ we built before the {\em first-level} data structure for a set of hyperplanes. Now suppose we already have the $i$-th level data structure for $i\geq 1$. Then, we build the $(i+1)$-th data structure on a set $A$ of hyperplanes as follows. We construct the arrangement $\calA$ of $A$ and build a point location data structure on $\calA$. Then, for each face $f$ of $\calA$, for the subset $A_f$ of all hyperplanes of $A$ above (resp., below) $f$, we build an $i$-th data structure on $A_f$. In this way, we build a $(d+1)$-th data structure on $A_v$, and we still use $\calD_v$ to denote the data structure. The total preprocessing time is $O(t^{d(d+1)})$. Doing this for all leaves of $T_D$ takes $t^{d(d+1)}\cdot 2^{O(t^d)}$ time and space, which is still $O(n)$.

In the same way as in the halfspace case, for each leaf cell $\Delta$ of our partition tree $T$, following the decision tree $T_D$, we find the leaf $v$ of $T_D$ that has the same configuration as $H(\Delta)$ and associate $v$ with $\Delta$. As before, doing this for all leaves of $T$ takes $O(n/t\cdot t^{d})$ time, which is $O(n\log n)$ if $\tau$ is small enough.

Given a query simplex $\sigma$, suppose we want to compute the number of points of $P(\Delta)$ inside $\sigma$ for a leaf cell $\Delta$ of $T$. In the dual setting, this is equivalent to computing the number of dual hyperplanes of $H(\Delta)$ that are in the ``correct'' sides of all $d+1$ dual points of the bounding hyperplanes of $\sigma$. This can be done in $O(\log t)$ time by using the data structure $\calD_v$. Indeed, let $p_i$ for $1\leq i\leq d+1$ be the dual point of the $i$-th bounding hyperplane of $\sigma$. Using the $(d+1)$-th level data structure of $\calD_v$, i.e., the point location data structure on $\calA_v$, we find the face $f$ of the arrangement $\calA(\Delta)$ of $H(\Delta)$ containing $p_{d+1}$ (again whenever the query algorithm attempts to use the $i$-th hyperplane of $A_v$, we use the $i$-th hyperplane of $H(\Delta)$ instead). Recall that we have two $d$-th level data structures associated with $f$, one on the set of hyperplanes above $f$ and the other below $f$; depending on the correct side of $p_{d+1}$, we recurse on the corresponding $d$-th data structure of $f$. At the first-level data structure, the face containing $p_1$ is associated with a number and we return the number as our answer to the simplex query (the face is actually associated with two numbers and we return the one corresponding to the ``correct'' side of $p_1$). In this way, the query can be answered after $d+1$ point locations, whose total time is $O(\log t)$ as $d=O(1)$.

In summary, with $O(n)$ space and $O(n\log n)$ time preprocessing, a simplex range counting query on $P(\Delta)$ for any leaf cell $\Delta$ of $T$ can be answered in $O(\log t)$ time.


\subsection{Reducing the preprocessing time}
\label{sec:simcountpretime}
We now reduce the preprocessing time of Lemma~\ref{lem:counting} to $O(n^{1+\epsilon})$ (with the same space and query time). The main idea is to build an ``upper partition tree'' of $O(1)$ depth using Matou\v{s}ek's method~\cite{ref:MatousekEf92} so that each leaf has $O(n^{\delta})$ points of $P$ for a small constant $\delta>0$ and then construct our data structure in Lemma~\ref{lem:counting} on each leaf (which form the ``lower'' part of the partition tree). The details are given below.

A {\em simplicial partition} for $P$ is a collection $\Pi = \{(P_1, \sigma_1), (P_2, \sigma_2), \ldots, (P_m, \sigma_m)\}$, where the $P_i$’s are pairwise disjoint subsets forming a partition of $P$, and each $\sigma_i$ is a simplex containing all points of $P_i$. The subsets $P_i$'s are called the {\em classes} and the simplices $\sigma_i$'s are called {\em cells}, which may overlap. The {\em crossing number} of $\Pi$ is the maximum number of cells crossed by any hyperplane.

\begin{lemma}\label{lem:buildsimpar}{\em(\cite{ref:MatousekEf92})}
For a set $P$ of $n$ points in $\bbR^d$ and a parameter $r\leq n^{\epsilon'}$ for any constant $\epsilon'<1$, a simplicial partition $\Pi = \{(P_1, \sigma_1), (P_2, \sigma_2), \ldots, (P_{O(r)}, \sigma_{O(r)})\}$
whose classes satisfy $|P_i|\leq n/r$ and whose crossing number is $O(r^{1-1/d})$ can be constructed in $O(n\log r)$ time.
\end{lemma}


We build a partition tree $T$ by Lemma~\ref{lem:buildsimpar} recursively, until we obtain a partition of $P$ into subsets of sizes $O(n^{\delta})$ for a small enough constant $\delta>0$ to be fixed later, which form the leaves of $T$. Each inner node $v$ of $T$ corresponds to a subset $P_v$ of $P$ as well as a simplicial partition $\Pi_v$ of $P_v$, which form the children of $v$. At each child $u$ of $v$, we store the cell $\sigma_u$ of $\Pi_v$ containing $P_u$ and also store the cardinality $|P_u|$. The simplicial partition $\Pi_v$ is constructed using Lemma~\ref{lem:buildsimpar} with parameter $r=n^{(1-\delta)/k}$ for a constant integer $k$ to be fixed later, i.e., every internal node of $T$ has $O(r)$ children.
If we recurse $k$ times, i.e., the depth of $T$ is $k$, then the subset size of each leaf of $T$ is $O(n^{\delta})$. Hence, the number of leaves of $T$ is $O(r^k)$.
Next, for each leaf $v$ of $T$, we construct the data structure of Lemma~\ref{lem:counting} on $P_v$, denoted by $\calD_v$, in $O(|P_v|^{d^2+2})$ time. Since $|P_v|=O(n^{\delta})$, we can make $\delta$ small enough so that the total time of Lemma~\ref{lem:counting} on all leaves of $T$ is $O(n^{1+\epsilon})$.
This finishes the preprocessing, which takes $O(n^{1+\epsilon})$ time. The space of the data structure is $O(n)$ since the depth of $T$ is $O(1)$.

Given a query simplex $\sigma$, starting from the root of $T$, for each node $v$ of $T$, we check whether $\sigma$ contains the cell $\sigma_v$. If yes, then we add $|P_v|$ to the total count. Otherwise, if the boundary of $\sigma$ crosses $\sigma_v$, then we proceed to the children of $v$.  In this way, we reach a set $V$ of leaves $v$ of $T$ whose cells $\sigma_v$ are crossed by the bounding hyperplanes of $\sigma$.
Since the number of leaves of $T$ is $O(r^k)$ and the depth of $T$ is $k$, which is a constant, the size of $V$ is $O(r^{k\cdot (1-1/d)})$.
If $Q(m)$ is the query time for a subset of size $m$, then we have the following recurrence
$$Q(m)=O(r)+O(r^{1-1/d})\cdot Q(m/r),$$
where $r=n^{(1-\delta)/k}$ with $n$ as the global input size.

Starting with $m=n$ and recursing $k$ times (using the same $r$) gives us
$$Q(n)=O(r^{(k-1)\cdot (1-1/d)+1})+O(r^{k\cdot (1-1/d)})\cdot Q(n/r^k).$$
Using  $r=n^{(1-\delta)/k}$, by setting $k$ to a constant integer larger than $(1/\delta -1)/(d-1)$, we obtain
\begin{equation}\label{equ:counting}
    Q(n)=O(n^{1-1/d-\delta'})+O(r^{k\cdot (1-1/d)})\cdot Q(n/r^k),
\end{equation}
for another small constant $\delta'>0$.

Finally, for each leaf node $v\in V$ (i.e., those subproblems $Q(n/r^k)$ in \eqref{equ:counting}), we use the data structure $\calD_v$ to compute the number of points of $P_v$ in $\sigma$, in $O(|P_v|^{1-1/d}/\log^{\Omega(1)}|P_v|)$ time by Lemma~\ref{lem:counting}, which is $O((n/r^k)^{1-1/d}/\log^{\Omega(1)}n)$ as $|P_v|=O(n/r^k)$. Plugging this into \eqref{equ:counting} gives us $Q(n)=O(n^{1-1/d}/\log^{\Omega(1)} n)$.
We thus conclude as follows.

\begin{theorem}\label{theo:counting}
Given a set of $n$ points in $\bbR^d$, there is a data structure of $O(n)$ space that can compute the number of points in any query simplex in $O(n^{1-1/d}/\log^{\Omega(1)} n)$ time. The data structure can be built in $O(n^{1+\epsilon})$ time for any $\epsilon>0$.
\end{theorem}

\section{Simplex range stabbing and segment intersection searching}
\label{sec:simstab}

Given a set $S$ of $n$ simplices in $\bbR^d$, the problem is to construct a data structure to compute the number of simplices that contain a query point (we also say that those simplices are {\em stabbed} by the query point).
In their data structure, Chan and Zheng~\cite{ref:ChanSi23} utilized a randomized  hierarchical partition. We follow their approach and instead use our deterministic hierarchical partition in Theorem~\ref{theo:partition}. Extra effort needs to be taken because the degree of the partition tree in \cite{ref:ChanOp12} is $O(1)$ while ours is logarithmic. In the following, we first briefly review Chan and Zheng's approach~\cite{ref:ChanSi23} and then discuss how to make changes.

\paragraph{A review of Chan and Zheng's randomized approach~\cite{ref:ChanSi23}.}
Each simplex is bounded by $d+1$ hyperplanes. A point $p$ stabs a simplex $s$ if $p$ is in the ``correct'' side of every bounding hyperplane of $s$ (to simplify the discussion, we assume these are lower halfplanes). In the dual space, this is equivalent to having the dual hyperplane of $p$ above the dual point of each bounding hyperplane of $s$. Therefore, we have the following problem in the dual space. Given a set $S^*$ of $n$ $(d+1)$-tuples of points in $\bbR^d$, we wish to construct a data structure to compute the number of tuples whose points are all below a query hyperplane $p^*$.
We can solve the problem using a multi-level data structure as follows.

Let $P$ be the set of the $(d+1)$-th points of all tuples of $S^*$.
Apply the simplicial partition of Lemma~\ref{lem:buildsimpar} to $P$ to obtain a partition $\Pi=\{(P_i,\sigma_i)\}$, with a parameter $r$ to be fixed later. Given a query hyperplane $p^*$,
for every cell $\sigma_i$ of $\Pi$ that is crossed by $p^*$, we recurse on the subset of tuples whose $(d+1)$-th points are in $P_i$ (i.e., apply Lemma~\ref{lem:buildsimpar} on $P_i$ recursively). For each cell $\sigma_i$ completely below $p^*$, we recurse on the subset of tuples whose $(d+1)$-th points are in $P_i$ but as a level-$d$ problem (the original problem is a level-$(d+1)$ problem; the recurrence stops at a level-$0$ problem in which we only need to return the size of the subset).
As analyzed in~\cite{ref:ChanOp12}, choosing $r=n^{\epsilon}$ with $n$ as the global input size (i.e., value $r$ is fixed for all levels of the recursion and this makes the recursion depth $O(1)$) can obtain an $O(n)$ space data structure with $O(n^{1-1/d+\epsilon})$ query time, for any $\epsilon>0$. The preprocessing time is $O(n\log n)$. We summarize this result in the lemma below (note that the result is deterministic).

\begin{lemma}\label{lem:basicstab}{\em (\cite{ref:ChanSi23})}
Given a set of $n$ simplices in $\bbR^d$, one can build a data structure of $O(n)$ space in $O(n\log n)$ time so that the number of simplices containing a query point can be computed in $O(n^{1-1/d+\epsilon})$ time, for any $\epsilon>0$.
\end{lemma}

To improve the query time, Chan and Zheng~\cite{ref:ChanSi23} reduced the problem to a special case where the first $d$ points in each tuple of $S^*$ all lie in $\bbR^{d-1}$ (equivalently, in the primal setting all but one bounding halfplanes of each input simplex are parallel to the $d$-th axis). As such, we can reduce some subproblems in the recurrence to the $(d-1)$-dimensional problem and then solve these subproblems by Lemma~\ref{lem:basicstab} (where $d$ becomes $d-1$). If $Q(n)$ (resp., $Q'(n)$) is the query time for the problem in $\bbR^d$ (resp., $\bbR^{d-1}$), then we have the following recurrence:
\begin{equation}\label{equ:stabquerytime}
Q(n) = O(r^{1-1/d})\cdot Q(n/r) + O(r)\cdot Q'(n/r).
\end{equation}
By Lemma~\ref{lem:basicstab}, $Q'(n)=O(n^{1-1/(d-1)+\epsilon})$. If we choose $r=n^{\delta}$ for a small constant $\delta>0$, the recursion depth is
$O(\log\log n)$ and thus the space is $O(n\log \log n)$. The query time is $O(n^{1-1/d}\log^{O(1)} n)$ due to a constant-factor blowup.

To further get rid of the logarithmic factor from the query time, a randomized hierarchical partition was used in~\cite{ref:ChanSi23} to replace all recursive simplicial partitions in the preprocessing algorithm for Lemma~\ref{lem:basicstab}. Here to achieve a deterministic result, we instead use our  deterministic hierarchical partition in Theorem~\ref{theo:partition} as follows.

\paragraph{Our deterministic result.}
With $r=n/b^c$ for a sufficiently large constant $c$ and $b=\log^{\rho}n$, we apply Theorem~\ref{theo:partition} to $P$ to obtain a partition tree $T$ consisting of $\Pi_j$, $0\leq j\leq k+1$. By Theorem~\ref{theo:partition}, $\Pi_j$ has $O(b'\cdot b^{j-1})$ cells and each cell has at most $2n/(b'\cdot b^{j-1})$ points of $P$. We define a sequence of numbers $t_0<t_1<\cdots <t_l$ as follows.
For each $i$, $0\leq i\leq l$, define $t_i=b'\cdot b^{j_i-1}$ (and thus $|\Pi_{j_i}|=O(t_i)$), with $j_0=1$ and
$$j_i=\left\lceil\log_b\frac{r^{1-(1-\epsilon)^i}}{b'}\right\rceil+1,\ \  1\leq i\leq l,$$ where $l$ is the smallest integer so that $r/b<b'\cdot b^{j_l-1}$. Hence, $l=O(\log\log r)$, which is $O(\log\log n)$. Note that
$b'\cdot b^{j_l-1}\leq r$ always holds.
By definition, $t_0=b'$ and $t_i$ is equal to $r^{1-(1-\epsilon)^i}$ within a factor of $b$ for $i\geq 1$.

We replace the recursive partitions in the preprocessing algorithm of Lemma~\ref{lem:basicstab} with the following sequence of partitions of $T$: $\Pi_0,\Pi_{j_0},\Pi_{j_1},\ldots,\Pi_{j_l}$. As $l=O(\log\log n)$, the space of the data structure is still $O(n\log \log n)$. Constructing $T$ takes  $O(n^{d^2+2})$ time by Theorem~\ref{theo:partition}.
The cells of the last partition $\Pi_{j_l}$ will be further preprocessed later.
The following lemma and corollary analyze the query time.

\begin{lemma}\label{lem:100}
The query time excluding the time spent on the cells of $\Pi_{j_l}$ is bounded by $O(b\cdot r^{1-1/d}\cdot (n/r)^{1-1/(d-1)+\epsilon})$.
\end{lemma}
\begin{proof}
Recall that for each $0\leq i\leq l$, the number of cells of $\Pi_{j_i}$ is $O(t_i)$. By Theorem~\ref{theo:partition}, the number of cells of $\Pi_{j_i}$ crossed by any hyperplane is $O(t_i^{1-1/d})$, the number of points in each cell of $\Pi_{j_i}$ is at most $O(n/t_i)$, and each cell of $\Pi_{j_i}$ contains at most $O(t_{i+1}/t_i)$ cells of $\Pi_{j_i+1}$.

Let $Q$ denote the query time excluding the time we spend on the cells of the last cutting $\Pi_{j_l}$.
For convenience, we further split the time of $Q$ into two portions: Let $Q_1$ be the portion of $Q$ spent on the cells of the partition $\Pi_{j_0}$; let $Q_2=Q-Q_1$. Following the above discussion, the time $Q_2$ is on the order of
\begin{equation}\label{equ:querybound}
\sum_{i=0}^{l-1}\frac{t_{i+1}}{t_i}\cdot t_i^{1-1/d}\cdot \left(\frac{n}{t_{i+1}}\right)^{1-1/(d-1)+\epsilon}
\end{equation}
Let $a=n/r$.
By definition, for each $0\leq i\leq l$, we have
$$r^{1-(1-\epsilon)^i}\leq t_i\leq  b\cdot r^{1-(1-\epsilon)^i}.$$
Using these, we obtain
\begin{equation*}
\begin{split}
& \frac{t_{i+1}}{t_i}\cdot t_i^{1-1/d}\cdot \left(\frac{n}{t_{i+1}}\right)^{1-1/(d-1)+\epsilon} = \frac{t_{i+1}}{t_i^{1/d}}\cdot \left(\frac{r\cdot a}{t_{i+1}}\right)^{1-1/(d-1)+\epsilon} \\
& \leq \frac{b\cdot r^{1-(1-\epsilon)^{i+1}}}{(r^{1-(1-\epsilon)^i})^{1/d}}\cdot \left(\frac{r\cdot a}{r^{1-(1-\epsilon)^{i+1}}}\right)^{1-1/(d-1)+\epsilon} \\
&= \frac{b\cdot r^{1-1/d}}{r^{(1-\epsilon)^{i+1}}\cdot r^{-1/d\cdot (1-\epsilon)^i}}\cdot   \left(\frac{a}{r^{-(1-\epsilon)^{i+1}}}\right)^{1-1/(d-1)+\epsilon}  \\
& = \frac{b\cdot a^{(1-1/(d-1)+\epsilon)}\cdot  r^{1-1/d}}{r^{(1-\epsilon)^{i+1}(1/(d-1)-\epsilon)}\cdot r^{-1/d\cdot (1-\epsilon)^i}} = \frac{b\cdot a^{(1-1/(d-1)+\epsilon)}\cdot  r^{1-1/d}}{r^{(1-\epsilon)^{i}\cdot (1/(d-1)-1/d-\Theta(\epsilon))}} \\
& = \frac{1}{r^{(1-\epsilon)^{i}\cdot (1/(d-1)-1/d-\Theta(\epsilon))}}\cdot b\cdot  a^{(1-1/(d-1)+\epsilon)}\cdot r^{1-1/d}\\
& = \frac{1}{r^{(1-\epsilon)^{i}\cdot (1/(d-1)-1/d-\Theta(\epsilon))}}\cdot
b \cdot \left(\frac{n}{r}\right)^{(1/(d-1)-1/d-\epsilon)}\cdot r^{1-1/d}.
\end{split}
\end{equation*}
Note that $\sum_{i=0}^l\frac{1}{r^{(1-\epsilon)^{i}\cdot (1/(d-1)-1/d-\Theta(\epsilon))}}=O(1)$. Hence, we obtain that \eqref{equ:querybound}, which is $Q_2$, is bounded by $O(b\cdot r^{1-1/d}\cdot (n/r)^{1-1/(d-1)+\epsilon})$.


For $Q_1$, since $j_0=1$ and $\Pi_1$ has $O(b')$ cells, each of which contains $O(n/b')$ points, $Q_1$ is bounded by $O(b'\cdot (n/b')^{1-1/(d-1)+\epsilon})$, which is also bounded by $O(b\cdot r^{1-1/d}\cdot (n/r)^{1-1/(d-1)+\epsilon})$ since $b'\leq b$. To see this, it suffices to show that $b'^{1/(d-1)-\epsilon}\leq b\cdot r^{1/(d-1)-1/d-\epsilon}$, which is apparently true since $b'<b$ and $r\geq 1$.

In summary, we obtain that $Q=Q_1+Q_2=O(b\cdot r^{1-1/d}\cdot (n/r)^{1-1/(d-1)+\epsilon})$. The lemma thus follows.
\end{proof}

\begin{corollary}\label{coro:10}
For a sufficiently large constant $c$, the query time excluding the time spent on the cells of $\Pi_{j_l}$ is bounded by $O(n^{1-1/d}/b)$.
\end{corollary}
\begin{proof}
In light of Lemma~\ref{lem:100}, it suffices to show that $b\cdot r^{1-1/d}\cdot (n/r)^{1-1/(d-1)+\epsilon}=O(n^{1-1/d}/b)$, which is equivalent to $b^2\cdot (n/r)^{1-1/(d-1)+\epsilon}=O((n/r)^{1-1/d})$. This in turn is equivalent to $b^2=O((n/r)^{1/(d-1)-1/d-\epsilon})$, which holds for a sufficiently large $c$ since $n/r=b^c$.
\end{proof}

\paragraph{Preprocessing cells of $\boldsymbol{\Pi_{j_l}}$.}
Recall that $\Pi_{j_l}$ has $O(b'\cdot b^{t_l-1})$ cells each containing $2n/(b'\cdot b^{t_l-1})$ points of $P$, and $r/b<b'\cdot b^{t_l-1}\leq r$.
Let $r'=b'\cdot b^{t_l-1}$.
By Corollary~\ref{coro:10}, we obtain the following recurrence on the query time $Q(n)$:
$$Q(n)=O(n^{1-1/d}/b)+O(r'^{1-1/d})\cdot Q_1(n/r'),$$
where $Q_1(n/r')$ is the query time for each subproblem of size $O(n/r')$ since each cell of $\Pi_{j_l}$ contains $O(n/r')$ points of $P$.
As $b=\log^{\rho}n$ and $r/b< r'\leq r$, we can write
\begin{equation}\label{equ:new50}
Q(n)=O(n^{1-1/d}/\log^{\Omega(1)}n)+O((n/n_1)^{1-1/d})\cdot Q_1(n_1),
\end{equation}
with $n_1=\log^{\Theta(1)}n$.

For notational convenience, let $T$ refer to the sequence of partitions $\Pi_0,\Pi_{j_0},\Pi_{j_1},\ldots,\Pi_{j_l}$, and  cells of $\Pi_{j_l}$ form the leaves of $T$.

To solve $Q_1(n_1)$ in \eqref{equ:new50}, we preprocess $P(\Delta)$ for each leaf cell $\Delta\in T$ recursively as above by using different parameters. Specifically, let $\tau>0$ be an arbitrarily small constant to be set later. Setting $r_1=n_1/\log^{\tau}n$, we apply Theorem~\ref{theo:partition} to construct a partition tree $T(\Delta)$ with $b_1=\log^{\rho}n_1$ in $O(n_1^{d^2+2})$ time.
As above (using the new parameters $r_1$ and $b_1$), we only use a subset of partitions of $T(\Delta)$ in our query data structure. For notational convenience, we use $T(\Delta)$ to refer to those partitions. The total time for constructing $T(\Delta)$ for all leaves $\Delta\in T$ is bounded by $O(n^{d^2+2})$ and the total space is still $O(n\log\log n)$.

With the same analysis as Lemma~\ref{lem:100}, we can obtain that the query time on $T(\Delta)$ excluding the time spent on the leaf cells of $T(\Delta)$ is $O(b_1\cdot r_1^{1-1/d}\cdot (n_1/r_1)^{1-1/(d-1)+\epsilon})$. Define $r_1'$ with respect to $r_1$ in the same way as above for $r'$ with respect to $r$, i.e., $T(\Delta)$ has $O(r_1')$ leaf cells and each cell contains $O(n_1/r_1')$ points of $P$. As above, $r_1/b_1<r_1'\leq r_1$ holds.
Consequently, we obtain the following
$$Q_1(n)=O(b_1\cdot r_1^{1-1/d}\cdot (n_1/r_1)^{1-1/(d-1)+\epsilon})+O(r_1'^{1-1/d})\cdot Q_2(n_1/r_1'),$$
where $Q_2(n_1/r_1')$ is the query time for each leaf cell of $T(\Delta)$. Let $t=n_1/r_1'$. We can now write
\begin{equation}\label{equ:new60}
Q_1(n)=O(b_1\cdot r_1^{1-1/d}\cdot (n_1/r_1)^{1-1/(d-1)+\epsilon})+O((n_1/t)^{1-1/d})\cdot Q_2(t).
\end{equation}

Combining \eqref{equ:new50} and \eqref{equ:new60} leads to
\begin{equation*}
    \begin{split}
        Q(n) & = O(n^{1-1/d}/\log^{\Omega(1)}n)+O((n/n_1)^{1-1/d})\cdot Q_1(n_1)\\
             & = O\left(\frac{n^{1-1/d}}{\log^{\Omega(1)}n}\right)+ O\left(\left(\frac{n}{n_1}\right)^{1-1/d}\cdot b_1\cdot r_1^{1-1/d}\cdot \left(\frac{n_1}{r_1}\right)^{1-1/(d-1)+\epsilon}\right) \\
             &  + O\left(\left(\frac{n}{n_1}\right)^{1-1/d}\cdot \left(\frac{n_1}{t}\right)^{1-1/d}\right)\cdot Q_2(t)\\
             & = O\left(\frac{n^{1-1/d}}{\log^{\Omega(1)}n}\right)+ O\left(n^{1-1/d}\cdot \left(\frac{r_1}{n_1}\right)^{1/(d-1)-1/d-\epsilon}\cdot b_1\right) + O\left(\left(\frac{n}{t}\right)^{1-1/d}\right)\cdot Q_2(t)\\
             & = O\left(\frac{n^{1-1/d}}{\log^{\Omega(1)}n}\right)+ O\left(\left(\frac{n^{1-1/d}}{\log^{\tau\cdot (1/(d-1)-1/d-\epsilon)} n}\right)\cdot \log^{\rho}\log n\right) + O\left(\left(\frac{n}{t}\right)^{1-1/d}\right)\cdot Q_2(t)\\
             & = O\left(\frac{n^{1-1/d}}{\log^{\Omega(1)}n}\right)+O\left(\left(\frac{n}{t}\right)^{1-1/d}\right)\cdot Q_2(t).\\
    \end{split}
\end{equation*}
Since $t=n_1/r_1'$ and $r_1/b_1< r_1'\leq r_1$, we have $n/t=n/n_1\cdot r_1'\leq n/n_1\cdot r_1=n/\log^\tau n$. Therefore, we further obtain from the above
\begin{equation}\label{equ:new70}
    \begin{split}
        Q(n) & = O\left(\frac{n^{1-1/d}}{\log^{\Omega(1)}n}\right)+ O\left(\left(\frac{n}{\log^{\tau}n}\right)^{1-1/d}\right)\cdot Q_2(t)\\
        & = O\left(\frac{n^{1-1/d}}{\log^{\Omega(1)}n}\right) + O\left(\frac{n^{1-1/d}}{\log^{\Omega(1)}n}\right)\cdot Q_2(t).
    \end{split}
\end{equation}
Since $t=n_1/r_1'$ and $r_1/b_1< r_1'\leq r_1$, we have $t\leq n_1/r_1\cdot b_1=O(\log^{\tau} n\cdot \log^{\rho}\log n)$. Therefore, we can write $t=O(\log^{\tau'}n)$ for another arbitrarily small constant $\tau'$.

In summary, the above first builds a partition tree $T$ and then builds partition trees $T(\Delta)$ for all leaf cells $\Delta$ of $T$ (using different parameters). For notational convenience, we still use $T$ to refer to the entire tree (i.e., attaching all leaf trees $T(\Delta)$ to $T$), which has $O(n/t)$ leaves, each containing $O(t)$ points. The space is bounded by $O(n\log\log n)$ and the total preprocessing time is $O(n^{d^2+2})$. In Section~\ref{sec:simcountsub}, by using the property that $t$ is very small, we show that after $O(n)$ space and $O(n\log n)$ time preprocessing, each query on any leaf cell of $T$ can be answered in $O(\log t)$ time (i.e., $Q_2(t)=O(\log\log n)$). As such, we obtain $Q(n)=O(n^{1-1/d}/\log^{\Omega(1)} n)$ from \eqref{equ:new70}.
We thus have the following result.

\begin{lemma}\label{lem:simstabcounting}
Given a set of $n$ simplices in $\bbR^d$, one can build a data structure of $O(n\log\log n)$ space so that the number of simplices containing a query point can be computed in $O(n^{1-1/d}/\log^{\Omega(1)} n)$ time. The data structure can be constructed in $O(n^{d^2+2})$ time.
\end{lemma}

We will further reduce the preprocessing time to $O(n^{1+\epsilon})$ in Section~\ref{sec:simstabpretime}.

\subsection{Solving the subproblems}
\label{sec:substabcount}

Let $A$ be a set of $t$ simplices in $\bbR^d$. We first build a basic data structure of $O(t^d)$ space in $O(t^{d+1})$ time that can answer each simplex stabbing counting query in $O(\log t)$ time.

Let $H$ be the set of the bounding hyperplanes of all simplices of $A$. Then, $|H|=(d+1)t$. We construct the arrangement $\calA$ of $H$ in $O(t^d)$ time and space~\cite{ref:EdelsbrunnerCo86}. For each face $f$ of $\calA$, we compute the number of simplices of $S$ that contain $f$, e.g., by checking every simplex of $A$; we associate that number with $f$, denoted by $N_f$. We also build a point location data structure on $\calA$~\cite{ref:ChazelleCu93}, which can be done in $O(t^d)$ time and space and support $O(\log t)$ time queries. This finishes our preprocessing, which takes $O(t^{d+1})$ time and uses $O(t^d)$ space.

Given a query point $p$, using the point location data structure we can find the face $f$ of $\calA$ that contains $p$ in $O(\log t)$ time, and then we return $N_f$ as the answer. The query time is $O(\log t)$.


Next, we show that after $O(n\log n)$ time and $O(n)$ space preprorcessing we can solve each subproblem $Q(t)$ in \eqref{equ:new70} in $O(\log t)$ time.

We build an algebraic decision tree $T_D$ for the arrangement construction algorithm~\cite{ref:EdelsbrunnerCo86} for a set of $(d+1)\cdot t$ hyperplanes. The height of $T_D$ is $O(t^d)$ and $T_D$ has $2^{O(t^d)}$ leaves. Each leaf of $T_D$ corresponds to a configuration of a set of $t$ simplices in the following sense (similar to those discussed in Section~\ref{sec:simcountsub}). Let $A$ and $B$ be two sets of $t$ simplices each. Let $H_A$ (resp., $H_B$) be the set of bounding hyperplanes of all simplices of $A$ (resp., $B$); we assign indices to hyperplanes of $H_A$ (resp., $H_B$) in such a way that indices of bounding hyperplanes from the same simplex are consecutive. Let $\calA_A$ and $\calA_B$ be the arrangements of $H_A$ and $H_B$, respectively. If the configurations of $H_A$ and $H_B$ are both the same as a leaf $v$ of $T_D$, then there is a one-to-one correspondence between faces of $\calA_A$ and $\calA_B$ such that
if a face $f$ of $\calA_A$ corresponds to a face $f'$ of $\calA_B$, then the $i$-th simplex of $A$ contains $f$ if and only if the $i$-th simplex of $B$ contains $f'$.

In light of the above observation, we do the following preprocessing.
For each leaf $v$ of $T_D$, let $A_v$ be a set of $t$ simplices with the same configuration as $v$. We construct the above basic data structure on $A_v$, denoted by $\calD_v$. Doing this for all leaves of $T_D$ takes $t^{d+1}\cdot 2^{O(t^d)}$ time and space, which is bounded by $O(n)$ as $t=\log^{\tau'} n$ for a small enough $\tau'$.

Recall that each subproblem $Q(t)$ in \eqref{equ:new70} is for the subset of tuples of $S^*$ whose $(d+1)$-th point is in a cell of $T$. For each cell $\sigma$ of $T$, let $S^*(\sigma)$ denote the corresponding subset of $S^*$ and let $S(\sigma)$ denote the subset of simplices of $S$ dual to the tuples of $S^*(\sigma)$. Recall that $\sigma$ contains at most $t$ points of $P$ (which is the set of $(d+1)$-th points of all tuples of $S^*$). Hence, $|S(\sigma)|\leq t$ (if $|S(\sigma)|<t$, we can add $t-|S(\sigma)|$ dummy simplices to $S(\sigma)$). Following the decision tree $T_D$, we determine in $O(t^d)$ time the leaf of $T_D$ whose configuration is the same as that of $S(\sigma)$; let $v_{\sigma}$ denote the leaf. The total time for determining $v_{\sigma}$ for all cells $\sigma$ of $T$ is $O(n/t\cdot t^d)=O(n\log n)$ time.

Given a query point $p$, suppose we want to compute the number of simplices of $S(\sigma)$ containing $p$ for a cell $\sigma$ of $T$. We use the data structure $\calD_{v_{\sigma}}$ and apply the query algorithm on $p$, but whenever the algorithm attempts to use the $i$-th simplex of $A_{v_{\sigma}}$ for computation we use the $i$-th simplex of $S(\sigma)$ instead. The query time is thus $O(\log t)$.


\subsection{Reducing the preprocessing time}
\label{sec:simstabpretime}

We now reduce the preprocessing time of Lemma~\ref{lem:simstabcounting} to $O(n^{1+\epsilon})$ using an idea similar to that for simplex range counting in Section~\ref{sec:simcountsub}. More specifically, we apply simplicial partitions of Matou\v{s}ek~\cite{ref:MatousekEf92} recursively but for only $O(1)$ recursive steps, so that each leaf has $O(n^{\delta})$ points of $P$ and then construct our data structure in Lemma~\ref{lem:simstabcounting} on each leaf (we make $\delta$ small enough so that the total preprocessing time on all leaves is bounded by $O(n^{1+\epsilon})$). The details are given below.

We compute a partition tree $T$ by constructing the simplicial partitions of Lemma~\ref{lem:buildsimpar} on $P$ recursively, until we obtain a partition of $P$ into subsets of sizes $O(n^{\delta})$ for a small enough constant $\delta>0$ to be fixed later, which form the leaves of $T$. Each inner node $v$ of $T$ corresponds to a subset $P_v$ of $P$ as well as a simplicial partition $\Pi_v$ of $P_v$, which form the children of $v$. We let $S^*_v$ denote the subset of tuples of $S^*$ whose $(d+1)$-th points are in $P_v$; let $S_v$ denote the subset of simplices of $S$ dual to the tuples of $S^*_v$.
At each child $u$ of $v$, we store the cell $\sigma_u$ of $\Pi_v$ containing $P_u$ and preprocess the corresponding tuple subset $S^*_u$ of $P_u$ using the algorithm of Lemma~\ref{lem:basicstab} for the $(d-1)$-dimensional problem. The simplicial partition $\Pi_v$ is constructed using Lemma~\ref{lem:buildsimpar} with parameter $r=n^{(1-\delta)/k}$ for a constant integer $k$ to be fixed later, i.e., each internal node of $T$ has $O(r)$ children.
If we recurse $k$ times, i.e., the depth of $T$ is $k$, then subset size of each leaf of $T$ is $O(n^{\delta})$. Hence, the number of leaves of $T$ is $O(r^k)$.
Next, for each leaf $v$ of $T$, we construct the data structure of Lemma~\ref{lem:simstabcounting} for $P_v$, denoted by $\calD_v$, in $O(|P_v|^{d^2+2})$ time. Since $|P_v|=O(n^{\delta})$, we can make $\delta$ small enough so that the total time of Lemma~\ref{lem:simstabcounting} on all leaves of $T$ is $O(n^{1+\epsilon})$.
This finishes the preprocessing, which takes $O(n^{1+\epsilon})$ time. The space of the data structure is $O(n)$ as $T$ has $O(1)$ levels.

Given a query point $p$, the query algorithm has two stages. In the first stage, we follow the same query algorithm as before using the partitions of $T$. We still have the recurrence \eqref{equ:stabquerytime} for the query time. The algorithm will find the set $V$ of leaves $v$ of $T$ whose cells $\sigma_v$ are crossed by the dual hyperplane $p^*$ of $p$.
Because the number of leaves of $T$ is $O(r^k)$ and the depth of $T$ is $k$, which is a constant, the size of $V$ is $O(r^{k\cdot (1-1/d)})$. If $Q(m)$ is the query time for a subset of size $m$, then we have the following recurrence $$Q(m)=O(r^{1-1/d})\cdot Q(m/r)+ O(r)\cdot Q'(m/r),$$
where $r=n^{(1-\delta)/k}$ and $Q'(m)=O(m^{1-1/(d-1)+\epsilon'})$ (here we use $\epsilon'$ instead for differentiation from the above $\epsilon$).

Starting with $m=n$ and recusing $k$ times gives us
$$Q(n)=O(r^{(k-1)\cdot (1-1/d)+1}\cdot (\frac{n}{r^k})^{1-1/(d-1)+\epsilon'})+O(r^{k\cdot (1-1/d)})\cdot Q(n/r^k).$$
Using  $r=n^{(1-\delta)/k}$ and setting $k$ to a constant integer larger than $(1/\delta -1)/(1/(d-1)-d\epsilon')$, we can derive
\begin{equation}\label{equ:stabbing}
    Q(n)=O(n^{(1-1/d)-\delta'})+O(r^{k\cdot (1-1/d)})\cdot Q(n/r^k),
\end{equation}
for another small constant $\delta'>0$.

Finally, for each leaf node $v\in V$ (i.e., those subproblems $Q(n/r^k)$ in the above recurrence~\eqref{equ:stabbing}), we use the data structure $\calD_v$ to solve the subproblem (i.e., computing the number of simplices of $S_v$ containing $p$), in $O(|P_v|^{1-1/d}/\log^{\Omega(1)}|P_v|)$ time by Lemma~\ref{lem:simstabcounting}, which is $O((n/r^k)^{1-1/d}/\log^{\Omega(1)} n)$ as $|P_v|=O(n/r^k)$. Plugging this into the above recurrence~\eqref{equ:stabbing} gives us $Q(n)=O(n^{1-1/d}/\log^{\Omega(1)} n)$.

\begin{theorem}\label{theo:simstabcounting}
Given a set of $n$ simplices in $\bbR^d$, one can build a data structure of $O(n\log\log n)$ space so that the number of simplices containing a query point can be computed in $O(n^{1-1/d}/\log^{\Omega} n)$ time. The data structure can be constructed in $O(n^{1+\epsilon})$ time for any $\epsilon>0$.
\end{theorem}


\subsection{Other related problems}
As studied in~\cite{ref:ChanSi23}, some related problems can be solved by similar techniques. For each of these problems, we basically follow the same high-level algorithmic framework as \cite{ref:ChanSi23} but use our deterministic partition tree instead of the randomized one in \cite{ref:ChanOp12}; this is very similar to the above simplex stabbing counting problem, so we will omit these details. For each problem, however, we still need to come up with a method to answer queries in $O(\log t)$ time for subproblems of tiny sizes $t=O(\log^{O(1)}\log n)$ and we briefly discuss this below.

\paragraph{Simplex range stabbing reporting.}
This is the reporting version of the above simplex stabbing counting problem, i.e., report all simplices of $S$ containing the query point.

We can design a basic data structure that is similar to that for the counting problem in Section~\ref{sec:substabcount}. One difference is that for each face $f$ of the arrangement $\calA$, we explicitly store the list of simplices that contain $f$. This increases the space of the data structure to $O(t^{d+1})$, which does not affect the space of the overall algorithm asymptotically. The rest of the algorithm (e.g., building the decision tree and associating each subproblem with a leaf of the decision tree) is the same as before. We thus have the following result.

\begin{theorem}\label{theo:simstabreport}
Given a set of simplices in $\bbR^d$, one can build a data structure of $O(n\log\log n)$ space so that the simplices containing a query point can be reported in $O(n^{1-1/d}/\log^{\Omega(1)} n+k)$ time, where $k$ is the output size. The data structure can be constructed in $O(n^{1+\epsilon})$ time for any $\epsilon>0$.
\end{theorem}

\paragraph{Segment intersection counting and reporting.}
Given a set $S$ of $n$ (possibly intersecting) line segments in the plane, the {\em segment intersection counting} problem is to construct a data structure to compute the number of segments of $S$ intersecting a query segment; the reporting problem is to report these segments.

Consider a subset $A$ of $t$ segments. We can have a basic data structure as follows. Let $P$ denote the set of the endpoints of all segments of $A$. Let $H$ be the set of dual lines of $P$. Let $\calA$ be the arrangement of $H$.
Each segment $s\in A$ is dual to a wedge~\cite{ref:deBergCo08} and a line $\ell$ crosses $s$ if and only if its dual point $\ell^*$ is in the dual wedge of $s$. The dual wedge of $s$ is bounded by two lines, which are dual to the endpoints of $A$. Hence, each face of $\calA$ is either completely inside or outside the dual wedge of $s$.
For each face of $f$ of $\calA$, we find the subset $A_f$ of segments of $S$ whose dual wedges contain $f$, e.g., by checking every segment of $A$. Suppose $s_q$ is a query segment and $\ell_q$ is the supporting line of $s_q$; let $\ell^*_q$ be the dual point of $\ell_q$. If $\ell^*_q$ is in a face $f$ of $\calA$, $A_f$ is exactly the subset of segments of $S$ intersecting $\ell_q$, and thus the segments of $A$ intersecting $s_q$ are all in $A_f$. To determine these segments, we further construct an arrangement of $H_f$, where $H_f$ is the set of supporting lines of the segments of $A_f$; let $\calA_f$ denote the arrangement. Let $f_1$ and $f_2$ be the two faces of $\calA_f$ containing the two endpoints of $s_q$, respectively. Then, a segment of $A_f$ intersects $s_q$ if and only its supporting line is between $f_1$ and $f_2$. Therefore, we do the following preprocessing. For each face $f_1$ of $\calA$, we build an array in which each element is indexed by a face $f_2\in \calA_f$ and the element stores the number of lines between $f_1$ and $f_2$ (for the reporting problem, we also store a list of segments of $A_f$ whose supporting lines are between $f_1$ and $f_2$), e.g., by checking every supporting line. The total preprocessing time and space can be easily bounded by $O(t^5)$. After the preprocessing, each segment intersection counting query can be answered in $O(\log t)$ time (there is an additional factor of the output size for a reporting query).

We build an algebraic decision tree $T_D$ for the above preprocessing algorithm (for a set of $t$ segments). The height of $T_D$ is $O(t^5)$ and $T_D$ has $2^{O(t^5)}$ leaves. Each leaf of $T_D$ corresponds to a configuration of the segments. For each leaf, we build a basic data structure as above on a set of $t$ segments with the same configuration as the leaf.
We then associate each subproblem of our original problem with its corresponding leaf of $T_D$ (i.e., the leaf has the same configuration as the subproblem). Given a query for a subproblem, we use the data structure built for the leaf associated with the subproblem to answer the query in $O(\log t)$ time.
The analysis is similar as before and we omit the details.

\begin{theorem}\label{theo:segquery}
Given a set of $n$ segments in the plane, one can build a data structure of $O(n\log\log n)$ space so that the number of segments intersecting a query segment can be computed in $O(\sqrt{n}/\log^{\Omega(1)} n)$ time (these segments can be reported in additional $O(k)$ time, where $k$ is the output size). The data structure can be constructed in $O(n^{1+\epsilon})$ time for any $\epsilon>0$.
\end{theorem}


\section{Segment intersection detection}
\label{sec:segdetect}

Let $S$ be a set of $n$ line segments in the plane. The problem is to build a data structure to determine whether a query line intersects any segment of $S$. Clearly, the problem can be solved by Theorem~\ref{theo:segquery}. In this section, we propose a different method that only needs $O(n)$ space while the query time is the same.

Let $P$ be the set of $2n$ endpoints of the segments of $S$.
With $r=4n/b^3$ and $b=\log^\rho n$, we apply Theorem~\ref{theo:partition} to $P$ to obtain a partition tree $T$ consisting of collections $\Pi_i$, $0\leq i\leq k+1$. By Theorem~\ref{theo:partition}, the total number of cells is $O(r)$, each cell in $\Pi_{k+1}$ contains at most $4n/r=\log^{3\rho} n$ points of $P$, and the runtime to construct $T$ is $O(n^{6})$. For each node $v\in T$, let $\Delta(v)$ denote the corresponding cell of $v$, which is a triangle in $\bbR^2$.

The following is a result from the previous work~\cite{ref:WangAl20} for a special case of the problem. We will use the result as a subroutine in our approach.

\begin{lemma}\label{lem:interdetectspecial}{\em (\cite{ref:WangAl20})}
If all segments of $S$ intersect a given line segment, then one can build a data structure of $O(n)$ space in $O(n \log n)$ time so that whether a query line intersects any segment of $S$ can be determined in $O(\log n)$ time.
\end{lemma}

We store the segments of $S$ in the partition tree $T$ as follows (the idea is similar to, but slightly different from that in \cite{ref:WangAl20}). For each segment $s\in S$, starting from the root, for each node $v$ whose cell $\Delta(v)$ contains $s$ (which is true initially when $v$ is the root), if $v$ is a leaf, then we store $s$ at $v$ (let $S_v$ denote the set of all such segments stored at $v$). Otherwise, we check every child of $v$. If $v$ has a child $u$ whose cell $\Delta(u)$ contains $s$, then we proceed on $u$. Otherwise, for each child $u$, if $\Delta(u)$ contains an endpoint of $s$, then since $\Delta(u)$ does not contain $s$, $s$ must intersect an edge $e$ of $\Delta(u)$; we store $s$ at $e$ (let $S_e$ denote the set of segments stored at $e$); note that since $s$ has two endpoints, there are two such edges $e$ but it suffices to store $s$ in one such edge. This finishes the algorithm for storing $s$, which takes $O(b\log n)$ time. Because $s$ is stored at either a leaf or a cell edge, the total space for storing all segments is $O(n)$. The total time is $O(nb\log n)$.

Next, for each edge $e$ of each cell of $T$, since all segments of $S_e$ intersect e, we preprocess $S_e$ using Lemma~\ref{lem:interdetectspecial}. Doing this for all cell edges $e$ of $T$ takes $O(n \log n)$ time and $O(n)$ space. For those segments stored in $S_v$ for all leaves $v$, we will preprocess them into a data structure $\calD_v$. Before discussing $\calD_v$, we describe the query algorithm and analyze the time complexity.

Given a query line $\ell$,
starting from the root of $T$, for each node $v$, assume that $\ell$ intersects the boundary of $\Delta(v)$, which is true initially when $v$ is the root. If $v$ is a leaf, then we call the data structure $\calD_v$ to check whether $\ell$ intersects a segment of $S_v$. Otherwise, for each child $u$ of $v$, for each edge $e$ of $\Delta(u)$, we apply the query algorithm of Lemma~\ref{lem:interdetectspecial} to check whether $\ell$ intersects a segment of $S_e$; further, if $\ell$ crosses $\Delta(u)$, then we proceed on $u$.

\begin{lemma}
The query algorithm works correctly.
\end{lemma}
\begin{proof}
If the query algorithm detects an intersection, then it is obviously true that $\ell$ intersects a segment of $S$. On the other hand, suppose $\ell$ intersects a segment $s$, say, at a point $p$. We argue that the query algorithm must detect an intersection. Indeed, according to our query algorithm, all nodes $u$ of $T$ whose cells $\Delta(u)$ are crossed by $\ell$ will be processed.
If $s$ is stored at a leaf $v$, then $s$ is contained in $\Delta(v)$. Since $\ell$ intersects $s$, $\ell$ must cross $\Delta(v)$, and thus $v$ must be processed and the data structure $\calD_v$ will detect an intersection between $\ell$ and $S_v$.

If $s$ is not stored at a leaf, then there must exist an internal node $v$ such that $s\in \Delta(v)$ and $s$ is not in $\Delta(u)$ for any child $u$ of $v$. According to our preprocessing algorithm, $s$ must be stored in $S_{e'}$ for an edge $e'$ of some cell $\Delta(u)$ of a child $u$ of $v$. Since $p\in s\subseteq \Delta(v)$ and $p\in \ell$, $\ell$ must cross $\Delta(v)$. Therefore, our query algorithm will process $v$ by applying the query algorithm of Lemma~\ref{lem:interdetectspecial} on $S_e$ for every edge $e$ of every child cell of $v$. When it is applied to $S_{e'}$, the intersection will be detected.
\end{proof}

We now analyze the query time. Recall that $T$ has $O(r)$ leaves. By Theorem~\ref{theo:partition}, the total number internal nodes of $T$ whose cell boundaries are crossed by $\ell$ is $O(\sqrt{r/b})$, and for each such node, we need to call the query algorithm of Lemma~\ref{lem:interdetectspecial} $O(b)$ times and each call takes $O(\log n)$ time. As such, the query time other than the time spent on calling $\calD_v$ for those leaves $v$ whose cell boundaries are crossed by $\ell$ (let $V$ be the set of all such cells) is
$O(\sqrt{r/b}\cdot b\cdot \log n)=O(\sqrt{r b}\cdot \log n)$.
By Theorem~\ref{theo:partition}, $|V|=O(\sqrt{r})=O(\sqrt{n/\log^{3\rho}n})$ and $|S_v|\leq 4n/r=\log^{3\rho}n$ for each leaf $v$ of $T$.

Let $Q(n)$ be the query time. Following the above analysis, we obtain the following
\begin{equation}\label{equ:new30}
Q(n)=O(\sqrt{rb}\cdot \log n) + O(\sqrt{r})\cdot Q_1(n/r),
\end{equation}
where $Q_1(\cdot)$ is the query time for each leaf $v\in V$, whose cell contains $O(n/r)$ points of $P$.

To solve $Q_1(n/r)$, we process $S_v$ for each leaf cell $v$ of $T$ as above by using different parameters. Specifically, let $n_1$ be the number of endpoints of $S_v$; hence $n_1=O(n/r)=O(\log^{3\rho}n)$. Let $\tau>0$ be an arbitrarily small constant to be set later. Setting $r_1=n_1/\log^{\tau}n$, we apply Theorem~\ref{theo:partition} to construct a partition tree $T(v)$ with $b_1=\log^{\rho}n_1$ in $O(n_1^{6})$ time. The total time for constructing $T(\Delta)$ for all leaf cells $\Delta$ of $T$ is thus bounded by $O(n^6)$.  Using $T(v)$ to handle queries on $S_v$ and following the above analysis, we obtain
\begin{equation}\label{equ:new40}
    Q(n_1)=O(\sqrt{r_1b_1}\cdot \log n_1) + O(\sqrt{r_1})\cdot Q_2(n_1/r_1),
\end{equation}
where $Q_2(\cdot)$ is the query time for each leaf cell of $T(v)$, which contains $O(n_1/r_1)=O(\log^{\tau} n)$ points of $P$.

Combining \eqref{equ:new30} and \eqref{equ:new40} leads to
\begin{equation}\label{equ:interdect}
    \begin{split}
        Q(n) & = O(\sqrt{rb}\cdot \log n) + O(\sqrt{r})\cdot Q_1(n/r)\\
             & = O(\sqrt{rb}\cdot \log n) + O(\sqrt{r}\cdot \sqrt{r_1b_1}\cdot \log n_1) + O(\sqrt{r}\cdot \sqrt{r_1})\cdot Q_2(n_1/r_1)\\
             & = O(\sqrt{rb}\cdot \log n) + O\left(\sqrt{r}\cdot \sqrt{\frac{nb_1}{r\log^{\tau}n}}\cdot \log n_1\right) + O\left(\sqrt{r}\cdot \sqrt{\frac{n}{r\log^{\tau}n}}\right)\cdot Q_2(n_1/r_1)\\
             & = O(\sqrt{rb}\cdot \log n) + O\left(\sqrt{\frac{nb_1}{\log^{\tau}n}}\cdot \log n_1\right) + O\left(\sqrt{\frac{n}{\log^{\tau}n}}\right)\cdot Q_2(n_1/r_1)\\
             & = O\left(\frac{\sqrt{n}}{b}\cdot \log n\right) + O\left(\sqrt{\frac{n}{\log^{\tau}n}}\cdot\sqrt{b_1}\cdot  \log n_1\right) + O\left(\sqrt{\frac{n}{\log^{\tau}n}}\right)\cdot Q_2(\log^{\tau} n)\\
             & = O\left(\frac{\sqrt{n}}{\log^{\rho-1} n}\right) + O\left(\sqrt{\frac{n}{\log^{\tau}n}}\cdot  \log^{\rho/2+1} \log n\right) + O\left(\sqrt{\frac{n}{\log^{\tau}n}}\right)\cdot Q_2(\log^{\tau} n),\\
             & = O\left(\frac{\sqrt{n}}{\log^{\Omega(1)} n}\right) + O\left(\sqrt{\frac{n}{t}}\right)\cdot Q_2(t),\\
    \end{split}
\end{equation}
where $t=\log^{\tau}n$.


In summary, the above first builds a partition tree $T$ and then builds partition trees $T(v)$ for all leaves $v$ of $T$ (using different parameters). For notational convenience, we use $T$ to refer to the entire tree (by attaching all trees $T(v)$ to $T$), which has $O(n/t)$ leaves, each containing $O(t)$ points. The space is bounded by $O(n)$.
As $t$ is small, we show below in Section~\ref{sec:subsinterdect} that with $O(n\log n)$ additional time and $O(n)$ space preprocessing, each subproblem $Q_2(t)$ in \eqref{equ:interdect} can be solved in $O(\log t)$ time. This makes the total query time $Q(n)$ bounded by $O(\sqrt{n}/\log^{\Omega(1)} n)$. We thus have the following result.

\begin{lemma}\label{lem:segdetect}
Given a set of $n$ segments in the plane, there is a data structure of $O(n)$ space that can determine whether a query line intersects any segment in $O(\sqrt{n}/\log^{\Omega(1)} n)$ time. The data structure can be built in $O(n^6)$ time.
\end{lemma}

We will further reduce the preprocesing time to $O(n^{1+\epsilon})$ in Section~\ref{sec:pretimesegdetect}.

\subsection{Solving subproblems}
\label{sec:subsinterdect}

Let $A$ be a set of $t$ line segments in the plane. We build a basic data structure to answer segment intersection detection queries as follows.
Let $P_A$ denote the set of the endpoints of all segments of $A$.
Let $H_A$ be the set of dual lines of the points of $P_A$. Let $\calA(H_A)$ be the arrangement of $H_A$. Each segment $s$ is dual to a wedge~\cite{ref:deBergCo08} and a line $\ell$ crosses $s$ if and only if its dual point $\ell^*$ is in the dual wedge of $s$. We first construct $\calA(H_A)$ in $O(t^2)$ time~\cite{ref:ChazelleTh85}. Then, for each face of $\calA(H_A)$, we mark it if it is contained in a dual wedge of a segment of $S$, e.g., by checking every segment of $S$.
We also build a point location data structure on $\calA(H_A)$ in $O(t^2)$ time and space~\cite{ref:EdelsbrunnerOp86,ref:KirkpatrickOp83,ref:SarnakPl86}.
This finishes the preprocessing, which takes $O(t^3)$ time and uses $O(t^2)$ space.
Given a query line $\ell$, using the point location data structure we find the face of $\calA(H_A)$ containing $\ell^*$ in $O(\log t)$ time and then check whether the cell is marked or not. As such, the query time is $O(\log t)$.

In what follows, we show that with $O(n\log n)$ time and $O(n)$ space preprocessing we can answer each query in $O(\log t)$ time on $S_v$ for each leaf $v$ of $T$.



For a set of $t$ segments, we build an algebraic decision tree $T_D$ for the preprocessing algorithm of the above basic data structure, excluding the step for constructing the point location data structure. The height of $T_D$ is $O(t^3)$ and $T_D$ has $2^{O(t^3)}$ leaves. Each leaf of $T_D$ corresponds to a configuration of a set of $t$ segments.
For each leaf $v\in T_D$, let $A_v$ be a set of $t$ segments whose configuration corresponds to $v$. We build the above basic data structure $\calD_v$ on $A_v$, which takes $O(t^3)$ preprocessing time and $O(t^2)$ space. Doing this for all leaves $v\in T_D$ takes a total of $t^3\cdot 2^{O(t^3)}$ time and space, which is bounded by $O(n)$ since $t=\log^{\tau} n$ for a small enough $\tau$.
For each leaf $v'$ of our partition tree $T$, we determine the leaf $v$ of $T_D$ that corresponds to the configuration of
$S_{v'}$, in $O(t^2)$ time using the decision tree $T_D$; we associate $v$ with $v'$. Doing this for $S_{v'}$ of all leaves $v'$ of $T$ takes $O(n/t\cdot t^3)$ time, which is $O(n\log n)$ since $t=\log^{\tau} n$ for a small enough $\tau$. This finishes our preprocessing, which takes $O(n \log n)$ time and uses $O(n)$ space.

Given a query line $\ell$, suppose we want to apply the query to $S_{v'}$ for a leaf $v'\in T$. Let $v$ be the leaf of $T_D$ associated with $v'$. We apply the query using the data structure $\calD_v$, but whenever the query algorithm attempts to use the $i$-th segment of $A_v$ we use the $i$-th segment of $S_{v'}$ instead. The query time is thus $O(\log t)$.

\subsection{Reducing the preprocessing time}
\label{sec:pretimesegdetect}

We reduce the preprocessing time of Lemma~\ref{lem:segdetect} to $O(n^{1+\epsilon})$.
Note that we cannot follow the same methods as in Sections~\ref{sec:simcountsub} or \ref{sec:substabcount} by constructing an upper partition tree using simplicial partitions of Lemma~\ref{lem:buildsimpar}. The reason is that cells in the simplicial partition may not be disjoint and consequently we cannot use the same approach as above to store segments of $S$.
Instead, we use an ``enhanced'' version of simplicial partitions proposed by Wang~\cite{ref:WangAl20},
which is the same as the original simplicial partition except that the partition $\Pi=\{(P_i,\sigma_i)\}$ has the following additional {\em weakly overlapped} property: For any cell $\sigma_i$, if $\sigma_i$ contains a point of $P_j$ for $i\neq j$, then all points of $P_i$ are outside $\sigma_j$.
Also, each cell $\sigma_i$ is now either a triangle or a convex quadrilateral.
The following result is from \cite{ref:WangAl20}.

\begin{lemma}\label{lem:enhance}{\em(\cite{ref:WangAl20})}
For a set $P$ of $n$ points in the plane and a parameter $r\leq n$, an enhanced simplicial partition $\Pi = \{(P_1, \sigma_1), (P_2, \sigma_2), \ldots, (P_{O(r)}, \sigma_{O(r)})\}$
whose classes satisfy $|P_i|\leq n/r$ and whose crossing number is $O(\sqrt{r})$ can be constructed in $O(nr + r^{5/2})$ time.
\end{lemma}

Recall that $P$ is the set of endpoints of all segments of $S$.
We start with constructing an enhanced simplicial partition of $P$: $\Pi=\{(P_1,\sigma_1),\cdots,(P_{O(r)},\sigma_{O(r)})\}$, in $O(nr+r^{5/2})$ time by Lemma~\ref{lem:enhance}, with $r$ to be fixed later.
For each segment $s\in S$, if both endpoints of $s$ are in the same class $P_i$ of $\Pi$, then we store $s$ at $\sigma_i$; let $S_i$ denote the subset of segments stored at $\sigma_i$. If the two endpoints of $s$ are in different classes, say, $P_i$ and $P_j$, then using the weakly overlapped property of $\Pi$ we can prove that $s$ must intersect an edge $e$ of either $\sigma_i$ or $\sigma_j$~\cite{ref:WangAl20}. We store $s$ at $e$ (if $s$ intersects multiple edges of the two cells, then we store $s$ at only one such edge); let $S_e$ be the subset of segments stored at $e$. We preprocess $S_e$ using Lemma~\ref{lem:interdetectspecial}. For each $S_i$, later we will build a data structure $\calD_i$ for it.
Assuming the space of $\calD_i$ is $O(|S_i|)$, the total space of the data structure is $O(n)$. The preprocessing time, excluding the time for constructing $\calD_i$'s, is $O(nr+r^{5/2}+n\log n)$.

Given a query line $\ell$, for each edge $e$ of each cell $\sigma_i$ of $\Pi$, we apply Lemma~\ref{lem:interdetectspecial} to check in $O(\log n)$ time whether $\ell$ intersects a segment of $S_e$. If yes, we can stop the query algorithm. Otherwise, for each cell $\sigma_i$ of $\Pi$ that is crossed by $\ell$, we use the data structure $\calD_i$ to check whether $\ell$ intersects any segment of $S_i$. If we construct $\calD_i$ recursively on $S_i$, then we have the following recurrence on the query time $Q(m)$ for a subset of size $m$:
\begin{equation}\label{equ:querysegdetect}
Q(m) = O(r\log m) + O(\sqrt{r})\cdot Q(m/r).
\end{equation}

We recurse $O(1)$ times and in each recursive step we use a fixed parameter $r$. 
These recursive partitions form a partition tree and we apply Lemma~\ref{lem:segdetect} on each leaf. The idea is similar to those before, e.g., in Section~\ref{sec:simcountsub}. The details are given below.

We build a partition tree $T$ recursively by Lemma~\ref{lem:enhance}, until we obtain a partition of $P$ into subsets of sizes $O(n^{\delta})$ for a small enough constant $\delta>0$ to be fixed later, which form the leaves of $T$. Each inner node $v$ of $T$ corresponds to a subset $P_v$ of $P$ as well as an enhanced simplicial partition $\Pi_v$ of $P_v$, which form the children of $v$. $\Pi_v$ is constructed by Lemma~\ref{lem:enhance} with parameter $r=n^{(1-\delta)/k}$ for a constant integer $k$ to be fixed later, i.e., each internal node of $T$ has $O(r)$ children.
If we recurse $k$ times, i.e., the depth of $T$ is $k$, then subset size of each leaf of $T$ is $O(n^{\delta})$. Hence, the number of leaves of $T$ is $O(r^k)$.
We store the segments of $S$ in $S_e$ for cell edges $e$ as well as in $S_v$ for leaves $v$ in the same way as discussed above. Next, for each leaf $v$ of $T$, we construct the data structure of Lemma~\ref{lem:segdetect} on $P_v$, denoted by $\calD_v$, in $O(|P_v|^{6})$ time. Since $|P_v|=O(n^{\delta})$, we can make $\delta$ small enough so that the total time of Lemma~\ref{lem:segdetect} on all leaves of $T$ is $O(n^{1+\epsilon})$.
This finishes the preprocessing.
The space of the data structure is $O(n)$ since the depth of $T$ is $O(1)$. For the preprocessing time, the time for constructing the data structure of Lemma~\ref{lem:interdetectspecial} on $S_e$ is $O(n\log n)$ as each segment of $S$ is stored in $S_e$ for at most one cell edge $e$.
For the time on constructing the enhanced simplicial partitions, observe that all classes in all partitions in the same level of $T$ form a partition of $P$, implying that the total size of these classes is $O(n)$. As such, since $T$ has $O(r^k)$ nodes and $r=n^{(1-\delta)/k}$, the total time for constructing the enhanced simplicial partitions is $O(nrk + r^{5/2}\cdot r^k)$, which is bounded by $O(n^{1+\epsilon})$ for $r=n^{\delta}$ if we make $k$ large enough.

For the query time, following recurrence \eqref{equ:querysegdetect},
starting with $m=n$ and recursing $k$ times gives us
$$Q(n)=O(r^{(k-1)/2+1}\cdot \log n)+O(r^{k/2})\cdot Q(n/r^k).$$
Using  $r=n^{(1-\delta)/k}$, by setting $k$ to a constant integer larger than $(1/\delta -1)/(d-1)$, we obtain
\begin{equation}\label{equ:querysegdetect10}
    Q(n)=O(n^{1/2-\delta'})+O(r^{k/2})\cdot Q(n/r^k),
\end{equation}
for another small constant $\delta'>0$.

Finally, for each leaf node $v$ corresponding to a subproblem $Q(n/r^k)$ in \eqref{equ:querysegdetect10}, we use the data structure $\calD_v$ to answer the query in $O(\sqrt{|P_v|}/\log^{\Omega(1)}|P_v|)$ time by Lemma~\ref{lem:interdetectspecial}, which is $O(\sqrt{n/r^k}/\log^{\Omega(1)} n)$ as $|P_v|=O(n/r^k)$. Plugging this into  recurrence~\eqref{equ:querysegdetect10} gives us $Q(n)=O(\sqrt{n}/\log^{\Omega(1)} n)$.
We thus conclude as follows.


\begin{theorem}
Given a set of $n$ segments in the plane, there is a data structure of $O(n)$ space that can determine whether a query line intersects any segment in $O(\sqrt{n}/\log^{\Omega(1)} n)$ time. The data structure can be built in $O(n^{1+\epsilon})$ time for any $\epsilon>0$.
\end{theorem}

\section{Ray-shooting among non-intersecting segments}
\label{sec:rayshoot}

Let $S$ be a set of $n$ line segments in the plane such that no two segments intersect. The problem is to build a data structure to compute the first segment hit by a query ray.

The following is a result from the previous work~\cite{ref:WangAl20} for a special case of the problem. We will use it as a subroutine in our approach.

\begin{lemma}\label{lem:rayshootspecial}{\em (\cite{ref:WangAl20})}
If all segments of $S$ intersect a given line segment, then one can build a data structure of $O(n)$ space in $O(n \log n)$ time so that a ray-shooting query can be answered in $O(\log n)$ time.
\end{lemma}

We build the partition tree $T$ and store $S$ in $T$ in the same way as for the segment intersection detection problem in Section~\ref{sec:segdetect} except that we use Lemma~\ref{lem:rayshootspecial} to preprocess $S_e$ for each cell edge $e$ of $T$. In addition, for each leaf $v$ of $T$, we will build a data structure $\calD_v$ on $S_v$.

Given a query ray $\rho$, starting from the root of $T$, for each node $v$, assume that $\rho$ intersects the boundary of $\Delta(v)$, which is true initially when $v$ is the root. If $v$ is a leaf, then we use the data structure $\calD_v$ to find the first segment of $S_v$ hit by $\rho$ as our candidate solution segment. Otherwise, for each child $u$ of $v$, for each edge $e$ of $\Delta(u)$, apply the query algorithm of Lemma~\ref{lem:rayshootspecial} to find the first ray of $S_e$ hit by $\rho$ as a candidate; further, if $\rho$ crosses $\Delta(u)$, then we proceed on $u$.
Finally, among all candidate segments, we return the one whose intersection with $\rho$ is closest to the origin of $\rho$.

\begin{lemma}
The query algorithm works correctly.
\end{lemma}
\begin{proof}
Suppose $s$ is the first segment of $S$ hit by $\rho$, say, at a point $p$. Note that $s$ is unique as all segments of $S$ are pairwise disjoint. We argue that the query algorithm will return $s$ as the solution. Indeed, observe that all candidate segments found by the query algorithm are hit by $\rho$. Hence, it suffices to show that $s$ is one of the candidate segments found by the query algorithm. Also observe that all nodes $v$ of $T$ whose cells $\Delta(v)$ are crossed by $\rho$ will be visited by the algorithm.

If $s$ is stored at a leaf $v$, then $p$ is in $\Delta(v)$ and thus $\rho$ must cross the boundary of $\Delta(v)$. Hence, our query algorithm will visit $v$ and thus call $\calD_v$ to find the first ray in $S_v$ hit by $\rho$ as a candiate solution. Since $s$ is the ray first hit by $\rho$ in $S$, $s$ must also be the ray first hit by $\rho$ in $S_v$. Hence, the data structure $\calD_v$ will return $s$ as a candidate solution.

If $s$ is not stored at a leaf, then there must be an internal node $v$ such that $s\in \Delta(v)$ and $s$ is not in $\Delta(u)$ for any child $u$ of $v$. According to our preprocessing algorithm, $s$ must be stored in $S_{e'}$ for an edge $e'$ of some cell $\Delta(u)$ of a child $u$ of $v$. Since $p\in s\subseteq \Delta(v)$ and $p\in \rho$, $\rho$ must cross the boundary of $\Delta(v)$. Therefore, our query algorithm will visit $v$ by applying the query algorithm of Lemma~\ref{lem:rayshootspecial} to $S_e$ for every edge $e$ of each child cell of $v$. When it is applied to $S_{e'}$, the segment $s$ will be found as a candidate segment.
\end{proof}

As in Section~\ref{sec:segdetect}, for each leaf $v$ of $T$, we construct the data structure $\calD_v$ by preprocessing $S_v$ recursively once.
The query time analysis follows exactly the same method as in Section~\ref{sec:segdetect} since the query time of Lemma~\ref{lem:rayshootspecial} is the same as that of Lemma~\ref{lem:interdetectspecial}, and thus we can also obtain the recurrence~\eqref{equ:interdect} with the same value of $t$. As $t$ is small, we show in Section~\ref{sec:subrayshoot} that with additional $O(n\log n)$ time and $O(n)$ space preprocessing, each subproblem $Q(t)$ in \eqref{equ:interdect} can be solved in $O(\log t)$ time. This makes the total query time $Q(n)$ bounded by $O(\sqrt{n}/\log^{\Omega(1)} n)$. We thus have the following result.

\begin{lemma}\label{lem:rayshoot}
Given a set of $n$ segments in the plane, there is a data structure of $O(n)$ space that can compute the first segment hit by a query ray in $O(\sqrt{n}/\log^{\Omega(1)} n)$ time. The data structure can be built in $O(n^6)$ time.
\end{lemma}

We will further reduce the preprocesing time to $O(n^{1+\epsilon})$ in Section~\ref{sec:pretimerayshoot}.

\subsection{Solving subproblems}
\label{sec:subrayshoot}

Let $A$ be a set of $t$ pairwise disjoint line segments in the plane. We build a basic data structure to answer ray-shooting queries as follows.
Let $P_A$ denote the set of the endpoints of all segments of $A$.
Let $H_A$ be the set of dual lines of the points of $P_A$. Let $\calA(H_A)$ be the arrangement of $H_A$.
After constructing the arrangement $\calA(H_A)$ and building a point location data structure on $\calA(H_A)$~\cite{ref:EdelsbrunnerOp86,ref:KirkpatrickOp83,ref:SarnakPl86}, for each face $F$ of $\calA(H_A)$, we find all segments of $A$ whose dual wedges contain $F$, e.g., by checking every segment of $A$; let $A_F$ denote the subset of such segments. We have the following observation.

\begin{observation}\label{obser:20}
Suppose $F$ is the face of $\calA(H_A)$ containing the dual point of the supporting line of a query ray $\rho$. Then, $A_F$ is the subset of segments of $A$ intersecting the supporting line of $\rho$; further, for each segment $s\in A_F$, $\rho$ hits the supporting line of $s$ if and only if $\rho$ hits $s$.
\end{observation}

For each face $F$ of $\calA(H_A)$, we compute the arrangement $\calA(A_F)$ of the supporting lines of the segments of $A_F$ and build a point location data structure on $\calA(A_F)$.
For each face $f$ of $\calA(A_F)$, which is a convex polygon, we store the edges of its boundary in a data structure (e.g., an array or a balanced binary search tree) so that given a ray with origin inside $f$, the edge of $f$ hit by the ray can be found in $O(\log |f|)$ time by binary search, where $|f|$ is the number of edges of $f$. Also, for each edge $e$ of $f$, we associate $s$ with $e$, where $s$ is the segment of $A_F$ whose supporting line contains $e$.

This finishes our preprocessing, which takes $O(t^4)$ time and space. Given a query ray $\rho$, we first locate the face $F$ of $\calA(H_A)$ containing the dual point of the supporting line of $\rho$. Then, we locate the face $f$ of $\calA(A_F)$ containing the origin of $\rho$. Finally, we find the edge $e$ of $f$ hits by the ray by binary search and return the segment of $S$ associated with $e$ as our answer to the ray-shooting query of $\rho$. The total query time is $O(\log t)$.

In what follows, we show that with $O(n\log n)$ time and $O(n)$ space preprocessing we can answer each query in $O(\log t)$ time on $S_v$ for each leaf $v$ of the partition tree $T$.

We construct the algebraic decision tree $T_D$ for the above preprocessing algorithm (excluding the steps for constructing the point location data structures; the steps for constructing the binary search data structures for the edges of faces $f$ of $\calA(A_F)$ are also not needed). The height of $T_D$ is $O(t^4)$ and $T_D$ has $2^{O(t^4)}$ leaves. Each leaf $v$ corresponds to a configuration of a set of $t$ segments. Let $A_v$ be a set of $t$ segments whose configuration is the same as $v$. We construct the above basic data structure $\calD_v$ on $A_v$. Doing this for all leaves takes $O(t^4\cdot 2^{O(t^4)})$ time and space, which is $O(n)$ since $t=\log^{\tau} n$ for a small enough $\tau$. For each leaf $v'$ of our partition tree $T$, following the decision tree $T_D$, we determine the leaf $v$ of $T_D$ that corresponds to the configuration of $S_{v'}$ and we associate $v$ with $v'$. This takes $O(t^4)$ time as the height of $T_D$ is $O(t^4)$. As $T$ has $O(n/t)$ leaves, doing this for all leaves of $T$ takes $O(n/t\cdot t^4)$ time, which is $O(n\log n)$ since $t=\log^{\tau} n$ for a small enough $\tau$.

Given a query ray $\rho$, suppose we want to find the first segment of $S_{v'}$ hit by $\rho$ for some leaf $v'$ of $T$. Then, let $v$ be the leaf of $T_D$ associated with $v'$. Using the data structure $\calD_v$ built on $A_v$, we apply the query algorithm of $\calD_v$, but whenever the query algorithm attempts to use the $i$-th segment of $A_v$, we use the $i$-th segment of $S_{v'}$ instead. The query time is thus $O(\log t)$.

\subsection{Reducing the preprocessing time}
\label{sec:pretimerayshoot}

We can basically follow the same approach as in Section~\ref{sec:pretimesegdetect} for the segment intersection detection problem. More specifically, we first build a partition tree $T$ using the enhanced simplicial partitions~\cite{ref:WangAl20}. Then we store segments of $S$ in $S_e$ for edges $e$ of cells of $T$ as well as in $S_v$ for leaves $v$ of $T$. We preprocess $S_e$ using Lemma~\ref{lem:rayshootspecial} and preprocess $S_v$ using Lemma~\ref{lem:rayshoot}. Following the same analysis as in Section~\ref{sec:pretimesegdetect}, the preprocessing time and space are $O(n^{1+\epsilon})$ and $O(n)$, respectively.
The query algorithm and time analysis are also similar. We conclude with the following theorem.

\begin{theorem}
Given a set of $n$ segments in the plane, there is a data structure of $O(n)$ space that can compute the first segment hit by a query ray in $O(\sqrt{n}/\log^{\Omega(1)} n)$ time. The data structure can be built in $O(n^{1+\epsilon})$ time for any $\epsilon>0$.
\end{theorem}

\paragraph{Acknowledgment.} The author would like to thank Timothy Chan for the discussions on derandomizing his partition tree.



 \bibliographystyle{plainurl}
\bibliography{reference}

\begin{thebibliography}{10}

\bibitem{ref:AgarwalRa17}
Pankaj~K. Agarwal.
\newblock {\em {\em Range searching, in} Handbook of Discrete and Computational
  Geometry, {\em C.D. T\'{o}th, J. O'Rourke, and J.E. Goodman (eds.)}}, pages
  1057--1092.
\newblock CRC Press, 3rd edition, 2017.

\bibitem{ref:AgarwalSi17}
Pankaj~K. Agarwal.
\newblock {\em {\em Simplex range searching and its variants: a review. In} A
  Journey Through Discrete Mathematics}, pages 1--30.
\newblock Springer, 2017.
\newblock \href {https://doi.org/10.1007/978-3-319-44479-6_1}
  {\path{doi:10.1007/978-3-319-44479-6_1}}.

\bibitem{ref:AgarwalRa93}
Pankaj~K. Agarwal and Ji\u{r}\'{i} Matou\v{s}ek.
\newblock Ray shooting and parametric search.
\newblock {\em SIAM Journal on Computing}, 22(4):794--806, 1993.
\newblock \href {https://doi.org/10.1137/0222051} {\path{doi:10.1137/0222051}}.

\bibitem{ref:AgarwalAp93}
Pankaj~K. Agarwal and Micha Sharir.
\newblock Applications of a new space-partitioning technique.
\newblock {\em Discrete and Computational Geometry}, 9:11--38, 1993.
\newblock \href {https://doi.org/10.1007/BF02189304}
  {\path{doi:10.1007/BF02189304}}.

\bibitem{ref:AgarwalPs05}
Pankaj~K. Agarwal and Micha Sharir.
\newblock Pseudoline arrangements: Duality, algorithms, and applications.
\newblock {\em SIAM Journal on Computing}, 34:526--552, 2005.
\newblock \href {https://doi.org/10.1137/S0097539703433900}
  {\path{doi:10.1137/S0097539703433900}}.

\bibitem{ref:Bar-YehudaVa94}
Reuven Bar-Yehuda and Sergio Fogel.
\newblock Variations on ray shootings.
\newblock {\em Algorithmica}, 11:133--145, 1994.
\newblock \href {https://doi.org/10.1007/BF01182772}
  {\path{doi:10.1007/BF01182772}}.

\bibitem{ref:deBergCo08}
Mark~de Berg, Otfried Cheong, Marc~J. van Kreveld, and Mark~H. Overmars.
\newblock {\em Computational Geometry --- Algorithms and Applications}.
\newblock Springer-Verlag, Berlin, 3rd edition, 2008.

\bibitem{ref:ChanOp12}
Timothy~M. Chan.
\newblock Optimal partition trees.
\newblock {\em Discrete and Computational Geometry}, 47:661--690, 2012.
\newblock \href {https://doi.org/10.1145/1810959.1810961}
  {\path{doi:10.1145/1810959.1810961}}.

\bibitem{ref:ChanHo23}
Timothy~M. Chan and Da~Wei Zheng.
\newblock Hopcroft's problem, log-star shaving, {2D} fractional cascading, and
  decision trees.
\newblock {\em ACM Transactions on Algorithms}, 2023.
\newblock \href {https://doi.org/10.1145/3591357} {\path{doi:10.1145/3591357}}.

\bibitem{ref:ChanSi23}
Timothy~M. Chan and Da~Wei Zheng.
\newblock Simplex range searching revisited: How to shave logs in multi-level
  data structures.
\newblock In {\em Proceedings of the 34th Annual ACM-SIAM Symposium on Discrete
  Algorithms (SODA)}, pages 1493--1511, 2023.
\newblock \href {https://doi.org/10.1137/1.9781611977554.ch54}
  {\path{doi:10.1137/1.9781611977554.ch54}}.

\bibitem{ref:ChazelleLo89}
Bernard Chazelle.
\newblock Lower bounds on the complexity of polytope range searching.
\newblock {\em Journal of the American Mathematical Society}, 2(4):637--666,
  1989.
\newblock \href {https://doi.org/10.2307/1990891} {\path{doi:10.2307/1990891}}.

\bibitem{ref:ChazelleCu93}
Bernard Chazelle.
\newblock Cutting hyperplanes for divide-and-conquer.
\newblock {\em Discrete and Computational Geometry}, 9(2):145--158, 1993.
\newblock \href {https://doi.org/10.1007/BF02189314}
  {\path{doi:10.1007/BF02189314}}.

\bibitem{ref:ChazelleTh85}
Bernard Chazelle, Leonidas~J. Guibas, and D.T. Lee.
\newblock The power of geometric duality.
\newblock {\em BIT}, 25:76--90, 1985.
\newblock \href {https://doi.org/10.1007/BF01934990}
  {\path{doi:10.1007/BF01934990}}.

\bibitem{ref:ChengAl92}
Siu~Wing Cheng and Ravi Janardan.
\newblock Algorithms for ray-shooting and intersection searching.
\newblock {\em Journal of Algorithms}, 13:670--692, 1992.
\newblock \href {https://doi.org/10.1016/0196-6774(92)90062-H}
  {\path{doi:10.1016/0196-6774(92)90062-H}}.

\bibitem{ref:EdelsbrunnerOp86}
Herbert Edelsbrunner, Leonidas~J. Guibas, and Jorge Stolfi.
\newblock Optimal point location in a monotone subdivision.
\newblock {\em SIAM Journal on Computing}, 15(2):317--340, 1986.
\newblock \href {https://doi.org/10.1137/0215023} {\path{doi:10.1137/0215023}}.

\bibitem{ref:EdelsbrunnerCo86}
Herbert Edelsbrunner, J.~O'Rourke, and Raimund Seidel.
\newblock Constructing arrangements of lines and hyperplanes with applications.
\newblock {\em SIAM Journal on Computing}, 15:341--363, 1986.
\newblock \href {https://doi.org/10.1137/0215024} {\path{doi:10.1137/0215024}}.

\bibitem{ref:EdelsbrunnerHa86}
Herbert Edelsbrunner and Emo Welzl.
\newblock Halfplanar range search in linear space and {$O(n^{0.695})$} query
  time.
\newblock {\em Information Processing Letters}, 23:289--293, 1986.
\newblock \href {https://doi.org/10.1016/0020-0190(86)90088-8}
  {\path{doi:10.1016/0020-0190(86)90088-8}}.

\bibitem{ref:GuibasIn88}
Leonidas~J. Guibas, Mark~H. Overmars, and Micha Sharir.
\newblock Intersecting line segments, ray shooting, and other applications of
  geometric partitioning techniques.
\newblock In {\em Proceedings of the 1st Scandinavian Workshop on Algorithm
  Theory (SWAT)}, pages 64--73, 1988.
\newblock \href {https://doi.org/10.1007/3-540-19487-8_7}
  {\path{doi:10.1007/3-540-19487-8_7}}.

\bibitem{ref:HausslerEp87}
David Haussler and Emo Welzl.
\newblock {$\epsilon$}-nets and simplex range queries.
\newblock {\em Discrete and Computational Geometry}, 2:127--151, 1987.
\newblock \href {https://doi.org/10.1007/BF02187876}
  {\path{doi:10.1007/BF02187876}}.

\bibitem{ref:KirkpatrickOp83}
David~G. Kirkpatrick.
\newblock Optimal search in planar subdivisions.
\newblock {\em SIAM Journal on Computing}, 12(1):28--35, 1983.
\newblock \href {https://doi.org/10.1137/0212002} {\path{doi:10.1137/0212002}}.

\bibitem{ref:MatousekEf92}
Ji\u{r}\'{i} Matou\v{s}ek.
\newblock Efficient partition trees.
\newblock {\em Discrete and Computational Geometry}, 8(3):315--334, 1992.
\newblock \href {https://doi.org/10.1007/BF02293051}
  {\path{doi:10.1007/BF02293051}}.

\bibitem{ref:MatousekRa93}
Ji\u{r}\'{i} Matou\v{s}ek.
\newblock Range searching with efficient hierarchical cuttings.
\newblock {\em Discrete and Computational Geometry}, 10(1):157--182, 1993.
\newblock \href {https://doi.org/10.1007/BF02573972}
  {\path{doi:10.1007/BF02573972}}.

\bibitem{ref:MatousekCu91}
Ji\v{r}\'{i} Matou\v{s}ek.
\newblock Cutting hyperplane arrangement.
\newblock {\em Discrete and Computational Geometry}, 6:385--406, 1991.
\newblock \href {https://doi.org/10.1007/BF02574697}
  {\path{doi:10.1007/BF02574697}}.

\bibitem{ref:MatousekGe94}
Ji\v{r}\'{i} Matou\v{s}ek.
\newblock Geometric range searching.
\newblock {\em ACM Computing Survey}, 26:421--461, 1994.
\newblock \href {https://doi.org/10.1145/197405.197408}
  {\path{doi:10.1145/197405.197408}}.

\bibitem{ref:OvermarsSt90}
Mark~H. Overmars, Haijo Schipper, and Micha Sharir.
\newblock Storing line segments in partition trees.
\newblock {\em BIT Numerical Mathematics}, 30:385--403, 1990.
\newblock \href {https://doi.org/10.1007/BF01931656}
  {\path{doi:10.1007/BF01931656}}.

\bibitem{ref:SarnakPl86}
Neil Sarnak and Robert~E. Tarjan.
\newblock Planar point location using persistent search trees.
\newblock {\em Communications of the ACM}, 29:669--679, 1986.
\newblock \href {https://doi.org/10.1145/6138.6151}
  {\path{doi:10.1145/6138.6151}}.

\bibitem{ref:WangAl20}
Haitao Wang.
\newblock Algorithms for subpath convex hull queries and ray-shooting among
  segments.
\newblock In {\em Proceedings of the 36th International Symposium on
  Computational Geometry (SoCG)}, pages 69:1--69:14, 2020.
\newblock \href {https://doi.org/10.4230/LIPIcs.SoCG.2020.69}
  {\path{doi:10.4230/LIPIcs.SoCG.2020.69}}.

\bibitem{ref:WangUn23}
Haitao Wang.
\newblock Unit-disk range searching and applications.
\newblock {\em Journal of Computational Geometry}, 14:343--394, 2023.
\newblock \href {https://doi.org/10.20382/jocg.v14i1a13}
  {\path{doi:10.20382/jocg.v14i1a13}}.

\bibitem{ref:WillardPo82}
Dan~E. Willard.
\newblock Polygon retrieval.
\newblock {\em SIAM Journal on Computing}, 11:149--165, 1982.
\newblock \href {https://doi.org/10.1137/0211012} {\path{doi:10.1137/0211012}}.

\bibitem{ref:YaoA83}
F.~Frances Yao.
\newblock A 3-space partition and its applications.
\newblock In {\em Proceedings of the 15th Annual ACM Symposium on Theory of
  Computing (STOC)}, pages 258--263, 1983.
\newblock \href {https://doi.org/10.1145/800061.808755}
  {\path{doi:10.1145/800061.808755}}.

\bibitem{ref:YaoPa89}
F.~Frances Yao, David~P. Dobkin, Herbert Edelsbrunner, and Mike Paterson.
\newblock Partitioning space for range queries.
\newblock {\em SIAM Journal on Computing}, 18:371--384, 1989.
\newblock \href {https://doi.org/10.1137/0218025} {\path{doi:10.1137/0218025}}.

\end{thebibliography}

%
%

\end{document}